\documentclass{article}
\usepackage{amsmath,amsfonts,amssymb,bm,cite,graphicx,algorithm,algorithmic,amsthm}
\usepackage{epsfig,subfigure}
\usepackage{fullpage,color}
\newtheorem{theorem}{Theorem}
\newtheorem{definition}{Definition}
\newtheorem{corollary}{Corollary}
\newtheorem{lemma}{Lemma}
\def\eop{{\hfill\vbox{\hrule height .3pt
      \hbox{\vrule width.3pt height 7pt
      \kern 7pt
      \vrule width .3pt}
      \hrule height .3pt}} \par\bigskip}
\begin{document}

\title{Robust recovery of complex exponential signals from random Gaussian projections via low rank Hankel matrix reconstruction}
\date{}
\author{
Jian-Feng Cai\thanks{Department of Mathematics, University of Iowa, Iowa City, IA 52242. Email: \texttt{\{jianfeng-cai,guibo-ye\}@uiowa.edu}}
\and
Xiaobo Qu\thanks{Department of Electronic Science, Fujian Provincial Key Laboratory of Plasma and Magnetic Resonance, State Key Laboratory of Physical Chemistry of Solid Surfaces, Xiamen University, P.O. Box 979, Xiamen 361005 (China). Email: \texttt{quxiaobo@xmu.edu.cn}}
\and
Weiyu Xu\thanks{Department of Electrical and Computer Engineering, University of Iowa, Iowa City, IA 52242. Email: \texttt{weiyu-xu@uiowa.edu}}
\and
Gui-Bo Ye$^*$
}

\maketitle

\begin{abstract}
This paper explores robust recovery of a superposition of $R$ distinct complex exponential functions from a few random Gaussian projections. We assume that the signal of interest is of $2N-1$ dimensional and $R<<2N-1$. This framework covers a large class of signals arising from real applications in biology, automation, imaging science, etc. To reconstruct such a signal, our algorithm is to seek a low-rank Hankel matrix of the signal by minimizing its nuclear norm subject to the consistency on the sampled data. Our theoretical results show that a robust recovery is possible as long as the number of projections exceeds $O(R\ln^2N)$. No incoherence or separation condition is required in our proof. Our method can be applied to spectral compressed sensing where the signal of interest is a superposition of $R$ complex sinusoids. Compared to existing results,  our result here does not need any separation condition on the frequencies, while achieving better or comparable bounds on the number of measurements. Furthermore, our method provides theoretical guidance on how many samples are required in the state-of-the-art non-uniform sampling in NMR spectroscopy. The performance of our algorithm is further demonstrated by numerical experiments.
\end{abstract}

\section{Introduction}

Many practical problems involve signals that can be modeled or approximated by a superposition of a few complex exponential functions. In particular, if we choose the exponential function to be complex sinusoid, it covers signals in acceleration of medical imaging \cite{LDP:MRM:07}, analog-to-digital conversion \cite{TLD:TIT:10}, inverse scattering in seismic imaging \cite{BPTP:IP:02}, etc.  Time domain signals in nuclear magnetic resonance (NMR) spectroscopy, that are widely used to analyze the compounds in chemistry and protein structures in biology, are another type of signals that can be modeled or approximated by a superposition of complex exponential functions \cite{QMCC:ACIE:15}. How to recover those superposition of complex exponential functions is of primary importance in those applications.

In this paper, we will consider how to recover those complex exponentials from linear measurements of their superposition. More specifically, let $\hat{\bm{x}}\in\mathbb{C}^{2N-1}$ be a vector satisfying
\begin{equation}\label{eq:hatx}
\hat{x}_j=\sum_{k=1}^{R}c_k z_k^j, \qquad j=0,1,\ldots,2N-2,
\end{equation}
where $z_k\in\mathbb{C}$, $k=1,\ldots,R$, are some unknown complex numbers. In other words, $\hat{\bm{x}}$ is a superposition of $R$ exponential functions. We assume $R\ll 2N-1$.  When $|z_k|=1$, $k=1,\ldots,R$, $\hat{\bm{x}}$ is a superposition of complex sinusoids. When $z_k=e^{-\tau_k}e^{2\pi\imath f_k}$, $k=1,\ldots,R$, $\hat{\bm{x}}$ models the signal in NMR spectroscopy.

Since $R\ll 2N-1$, the degree of freedom to determine $\hat{\bm{x}}$ is much less than the ambient dimension $2N-1$. Therefore, it is possible to recover $\hat{\bm{x}}$ from its under-sampling \cite{CLMW:JACM:11,CP:ProcIEEE:09,CRT:TIT:06,Don:TIT:06}.
In particular, we consider to recover $\hat{\bm{x}}$ from its linear measurement
\begin{equation}\label{eq:linmea}
\bm{b}=\mathcal{A}\hat{\bm{x}},
\end{equation}
where $\mathcal{A}\in\mathbb{C}^{M\times (2N-1)}$ with $M\ll 2N-1$.

We will use a Hankel structure to reconstruct the signal of interest $\hat{\bm{x}}$. The Hankel structure originates from the matrix pencil method \cite{HS:TASSP:90} for harmonic retrieval for complex sinusoid. The conventional matrix pencil method assumes fully observed $\hat{\bm{x}}$ as well as the model order $R$, which are both unknown here. Following the ideas of the matrix pencil method in \cite{HS:TASSP:90} and enhanced matrix completion (EMaC) in \cite{CC:TIT:14}, we construct a Hankel matrix based on signal $\hat{\bm{x}}$. \ More specifically, define the Hankel matrix $\hat{\bm{H}}\in\mathbb{C}^{N\times N}$ by
\begin{equation}\label{eq:Hankel}
\hat{H}_{jk}=\hat{x}_{j+k},\qquad j,k=0,1,\ldots,N-1.
\end{equation}
Throughout this paper, indices of all vectors and matrices start from $0$, instead of $1$ in conventional notations. It can be shown that $\hat{\bm{H}}$ is a matrix with rank $R$. Instead of reconstructing $\hat{\bm{x}}$ directly, we reconstruct the rank-$R$ Hankel matrix $\hat{\bm{H}}$, subject to the constraint that \eqref{eq:linmea} is satisfied.

Low rank matrix recovery has been widely studied  \cite{CCS:SIOPT:10,CP:ProcIEEE:09,CR:FoCM:09,RFP:SIREV:10}. It is well known that minimizing the nuclear norm tends to lead to a solution of low-rank matrices. Therefore,  a nuclear norm minimization problem subject to the constraint \eqref{eq:linmea} is proposed. More specifically,
for any given $\bm{x}\in{\mathbb{C}^{2N-1}}$, let $\bm{H}(\bm{x})\in\mathbb{C}^{N\times N}$ be the Hankel matrix whose first row and last column is $\bm{x}$, i.e., $[\bm{H}(\bm{x})]_{jk}=x_{j+k}$. We propose to solve
\begin{equation}\label{eq:min}
\min_{\bm{x}}\|\bm{H}(\bm{x})\|_*,\qquad\mbox{subject to}\quad \mathcal{A}\bm{x}=\bm{b},
\end{equation}
where $\|\cdot\|_*$ is the nuclear norm function (the sum of all singular values), and $\mathcal{A}$ and $\bm{b}$ are from the linear measurement \eqref{eq:linmea}. When there is noise contained in the observation, i.e.,
$$
\bm{b}=\mathcal{A}\hat{\bm{x}}+\bm{\eta},
$$
we solve
\begin{equation}\label{eq:minnoise}
\min_{\bm{x}}\|\bm{H}(\bm{x})\|_*,\qquad\mbox{subject to}\quad \|\mathcal{A}\bm{x}-\bm{b}\|_2\leq\delta,
\end{equation}
where $\delta=\|\bm{\eta}\|_2$ is the noise level.

An important theoretical question is how many measurements are required to get a robust reconstruction of $\hat{\bm{H}}$ via \eqref{eq:min} or \eqref{eq:minnoise}. For a generic unstructured $N\times N$ matrix of rank $R$, standard theory \cite{CR:FoCM:09,CT:TIT:10,CRPW:FoCM:12,RFP:SIREV:10} indicates that $O(NR\cdot poly(\log N))$ measurements are needed for a robust reconstruction by nuclear norm minimization. This result, however, is unacceptable here since the number of parameters of $\hat{\bm{H}}$ is only $2N-1$. The main contribution of this paper is then to prove that \eqref{eq:min} and \eqref{eq:minnoise} give a robust recovery of $\hat{\bm{H}}$ (hence $\hat{\bm{x}}$) as soon as the number of projections exceeds $O(R\ln^2N)$ if we choose the linear operator $\mathcal{A}$ to be some scaled random Gaussian projections. This result is further extended to the robust reconstruction of low-rank Hankel or Toeplitz matrices from its few Gaussian random projections.

Our result can be applied to various signals of superposition of complex exponentials, including, but not limited to, signals of complex sinusoids and signals in accelerated NMR spectroscopy. When applied to complex sinusoids, our result here does not need any separation condition on the frequencies, while achieving better or comparable bounds on the number of required measurements.
Furthermore, 
our theoretical result provides some guidance on how many samples to choose for the model proposed in \cite{QMCC:ACIE:15} to recover NMR spectroscopy.

\begin{itemize}
\item {\bf Complex sinusoids.} When $|z_k|=1$ for $k=1,\ldots,R$, we must have $z_k=e^{2\pi \imath f_k}$ for some frequency $f_k$. In this case, $\hat{\bm{x}}$ is a superposition of complex sinusoids, for examples, in the analog-to-digital conversion of radio signals \cite{TLD:TIT:10}. The problem on recovering  $\hat{\bm{x}}$ from its as few as possible linear measurements \eqref{eq:linmea} may be solved using compressed sensing (CS)\cite{CRT:TIT:06}. One can discretize the domain of frequencies $f_k$ by a uniform grid. When the frequencies $f_k$ indeed fall on the grid, $\hat{\bm{x}}$ is sparse in the discrete Fourier transform domain, and CS theory \cite{CRT:TIT:06,Don:TIT:06} suggests that it is possible to reconstruct $\hat{\bm{x}}$ from its very few samples via $\ell_1$-norm minimization, provided that $R\ll 2N-1$. Nevertheless, the frequencies $f_k$ in our setting usually do not exactly fall on a grid. The basis mismatch between the true parameters and the grid based on discretization degenerates the performance of conventional compressed sensing \cite{CLPCR:SP:11}.

    To overcome this, \cite{CF:CPAM:12,TBSR:TIT:13} proposed to recover off-the-grid complex sinusoid frequencies using total variation minimization or atomic norm \cite{CRPW:FoCM:12} minimization. They proved that the total variation minimization or atomic norm minimization can have a robust reconstruction of $\hat{\bm{x}}$ from a nonuniform sampling of very few entries of $\hat{\bm{x}}$, provided that the frequencies $f_k$, $k=1,\ldots,R$, has a good separation. Another method for recovering off-the-grid frequencies is enhanced matrix completion (EMaC) proposed by Chen et al \cite{CC:TIT:14}, where the Hankel structure plays a central role similar to our model. The main result in \cite{CC:TIT:14} is that the complex sinusoids $\bm{\hat{x}}$ can be robustly reconstructed via EMaC from its very few nonuniformly sampled entries. Again, the EMaC requires a separation of the frequencies, described implicitly by an incoherence condition.

    When applied to complex sinusoids, compared to the aforementioned existing results, our result here does not need any separation condition on the frequencies, while achieving better or comparable bound of number of measurements.

\item {\bf Accelerated NMR spectroscopy.} When $z_k=e^{-\tau_k}e^{2\pi\imath f_k}$, $k=1,\ldots,R$, $\hat{\bm{x}}$ models the signal in NMR spectroscopy, which arises frequently in studying short-lived molecular systems, monitoring chemical reactions in real-time, high-throughput applications, etc. Recently, Qu et al \cite{QMCC:ACIE:15} proposed an algorithm based on low rank Hankel matrix. In this specific application, $\mathcal{A}$ is a matrix that denotes the under-sampling of NMR signals in the time domain. Numerical results show its efficiency in \cite{QMCC:ACIE:15} while theoretical results are still needed to explain. It is vital to give some theoretical results on this model since it will give us some guidance on how many samples should be chosen to guarantee the robust recovery. Though the result in \cite{CC:TIT:14} applies to this problem, it needs an incoherence condition, which remains uncertain for diverse chemical and biology samples. Our result in this paper does not require any incoherence condition. Moreover, our bound is better than that in \cite{CC:TIT:14}.
\end{itemize}

The rest of this paper is organized as follows. We begin with our model and our main results in Section \ref{secMMR}. Proofs for the main result are given in Section \ref{secProof}. Then, in Section \ref{secMatrices}, we  extend the main result to the reconstruction of generic low-rank Hankel or Toeplitz matrices. Finally, the performance of our algorithm is demonstrated by numerical experiments in Section \ref{secNum}.

\section{Model and Main Results}\label{secMMR}
Our approach is based on the observation that the Hankel matrix whose first row and last column consist of entries of $\hat{\bm{x}}$ has rank $R$.  Let $\hat{\bm{H}}$ be the Hankel matrix defined by \eqref{eq:Hankel}. Eq. \eqref{eq:hatx} leads to a decomposition
$$
\hat{\bm{H}}=
\left[
\begin{matrix}
1&\ldots&1\cr
z_1&\ldots&z_R\cr
\vdots&\vdots&\vdots\cr
z_1^{N-1}&\ldots&z_R^{N-1}\cr
\end{matrix}
\right]
\left[\begin{matrix}
c_1\cr &\ddots\cr&&c_R
\end{matrix}
\right]
\left[
\begin{matrix}
1&z_1\ldots&z_1^{N-1}\cr
\vdots&\vdots&\vdots\cr
1&z_R\ldots&z_R^{N-1}\cr
\end{matrix}
\right]
$$
Therefore, the rank of $\hat{\bm{H}}$ is $R$. Similar to Enhanced Matrix Completion (EMaC) in \cite{CC:TIT:14}, in order to reconstruct $\hat{\bm{x}}$, we first reconstruct the rank-$R$ Hankel matrix $\hat{\bm{H}}$, subject to the constraint that \eqref{eq:linmea} is satisfied. Then, $\hat{\bm{x}}$ is derived directly by choosing the  first row and last column of $\hat{\bm{H}}$. More specifically,
for any given $\bm{x}\in{\mathbb{C}^{2N-1}}$, let $\bm{H}(\bm{x})\in\mathbb{C}^{N\times N}$ be the Hankel matrix whose first row and last column is $\bm{x}$, i.e., $[\bm{H}(\bm{x})]_{jk}=x_{j+k}$. We propose to solve
\begin{equation}\label{eq:minrank}
\min_{\bm{x}}\mathrm{rank}(\bm{H}(\bm{x})),\qquad\mbox{subject to}\quad \mathcal{A}\bm{x}=\bm{b},
\end{equation}
where $\mathrm{rank}(\bm{H}(\bm{x}))$ denotes the rank of $\bm{H}(\bm{x})$, and $\mathcal{A}$ and $\bm{b}$ are from the linear measurement \eqref{eq:linmea}. When there is noise contained in the observation, i.e, $\bm{b}=\mathcal{A}\hat{\bm{x}}+\eta$, we correspondingly solve

\begin{equation}\label{eq:minranknoise}
\min_{\bm{x}}\mathrm{rank}(\bm{H}(\bm{x})),\qquad\mbox{subject to}\quad \|\mathcal{A}\bm{x}-\bm{b}\|_2\leq \delta,
\end{equation}
where $\delta=\|\eta\|_2$ is the noise level.

These two problems are all NP hard problems and not easy to solve. Following the ideas of matrix completion and low rank matrix recovery \cite{CR:FoCM:09,CT:TIT:10,CRPW:FoCM:12,RFP:SIREV:10}, it is possible to exactly recover the low rank Hankel matrix via nuclear norm minimization. Therefore, it is reasonable to use nuclear norm minimization for our problem and it leads to the models in \eqref{eq:min} and \eqref{eq:minnoise}.

Intuitively, our model is reasonable and likely to work. Theoretical results are desirable to guarantee it. The results in \cite{CR:FoCM:09,CT:TIT:10,CRPW:FoCM:12,RFP:SIREV:10} do not consider the Hankel structure. For generic $N\times N$ rank-$R$ matrix, they requires $O(NR\cdot poly(\log N))$ measurements for robust recovery which is too much since there are only $2N-1$ degrees of freedom in $\bm{H}(\bm{x})$. The theorems proposed in \cite{TBSR:TIT:13} work only for a special case where signals of interest are superpositions of complex sinusoids, which excludes, e.g., the signals in NMR spectroscopy. While the results from \cite{CC:TIT:14} extend to complex exponentials, the performance guarantees in \cite{TBSR:TIT:13,CC:TIT:14,CF:CPAM:12} require incoherence conditions, implying the knowledge of frequency interval in spectroscopy, which are not available before the realistic sampling of diverse chemical or biological samples. This limits the applicability of these theories.

It is challenging to provide a theorem guaranteeing the exact recovery for model \eqref{eq:min} with arbitrarily linear measurements $\mathcal{A}$. In this paper, we provide a theoretical result ensuring exact recovery when $\mathcal{A}$ is a scaled random Gaussian matrix. Our result does not assume any incoherence conditions on the original signal.

\begin{theorem}\label{thm:main}
Let $\mathcal{A}=\mathcal{B}\mathcal{D}\in\mathbb{C}^{M\times (2N-1)}$, where $\mathcal{B}\in\mathbb{C}^{M\times (2N-1)}$ is a random matrix whose real and imaginary parts are i.i.d. Gaussian with mean $0$ and variance $1$, $\mathcal{D}\in\mathbb{R}^{(2N-1)\times (2N-1)}$ is a diagonal matrix with the $j$-th diagonal $\sqrt{j+1}$ if $j\leq N-1$ and $\sqrt{2N-1-j}$ otherwise. Then, there exists a universal constant  $C_1>0$ such that, for an arbitrary $\epsilon>0$,  If
$$
M \geq (C_1\sqrt{R}\ln N+\sqrt{2}\epsilon)^2+1,
$$
then, with probability at least $1-2e^{-\frac{M-1}{8}}$, we have
\begin{enumerate}
\item[(a)]
$\tilde{\bm{x}}=\hat{\bm{x}}$, where $\tilde{\bm{x}}$ is the unique solution of \eqref{eq:min} with $\bm{b}=\mathcal{A}\hat{\bm{x}}$;
\item[(b)]
$\|\mathcal{D}(\tilde{\bm{x}}-\hat{\bm{x}})\|_2\leq 2\delta/\epsilon$, where $\tilde{\bm{x}}$ is the unique solution of \eqref{eq:minnoise} with $\|\bm{b}-\mathcal{A}\hat{\bm{x}}\|_2\leq\delta$.
\end{enumerate}
\end{theorem}

The number of measurements required is $O(R\ln^2N)$, which is reasonable small compared with the number of parameters in $\bm{H}(\bm{x})$. Furthermore, there is a parameter $\epsilon$ in Theorem \ref{thm:main}. For the noise-free case (a), the best choice of $\epsilon$ is obviously a number that is very close to $0$. For the noisy case (b), we can balance the error bound and the number of measurements to get an optimal $\epsilon$. On the one hand, according to the result in (b), in order to make the error in noisy case as small as possible, we would like $\epsilon$ to be as large as possible. On the other hand, we would like to keep the measurements $M$ of the order of $R\ln^2N$. Therefore, a seemingly optimal choice of $\epsilon$ is $\epsilon=O(R\ln^2N)$. With this choice of $\epsilon$, the number of measurements $M=O(R\ln^2N)$ and the error $\|\mathcal{D}(\tilde{\bm{x}}-\hat{\bm{x}})\|_2\leq O\left(\frac{\delta}{\sqrt{M}}\right)$.

\section{Proof of Theorem \ref{thm:main}} \label{secProof}
In this section, we prove the main result Theorem \ref{thm:main}.

\subsection{Orthonormal Basis of the $N\times N$ Hankel Matrices Subspace}
In this subsection, we introduce an orthonormal basis of the subspace of $N\times N$ Hankel matrices and use it to define a projection from $\mathbb{C}^{N\times N}$ to the subspace of all $N\times N$ Hankel matrices.

 Let $\bm{E}_j\in\mathbb{C}^{N\times N}$, $j=0,1\ldots,2N-2$, be the Hankel matrix satisfying
\begin{equation}\label{eq:Ej}
[\bm{E}_j]_{k l}=
\begin{cases}
1/\sqrt{K_j},&\mbox{if }k+l=j,\cr
0,&\mbox{otherwise,}
\end{cases}
\qquad
k,l=0,\ldots,N-1,
\end{equation}
where $K_j=j+1$ for $j\leq N-1$ and $K_j=2N-1-j$ for $j\geq N-1$ is the number of non-zeros in $\bm{E}_j$. Then, it is easy to check that $\{\bm{E}_j\}_{j=0}^{2N-2}$ forms an orthonormal basis of the subspace of all $N\times N$ Hankel matrices, under the standard inner product in $\mathbb{C}^{N\times N}$.

Define a linear operator
\begin{equation}\label{def:G}
\mathcal{G}~:~\bm{x}\in\mathbb{C}^{2N-1}\mapsto\mathcal{G}\bm{x}=\sum_{j=0}^{2N-2}x_j\bm{E}_j\in\mathbb{C}^{N\times N}.
\end{equation}
The adjoint $\mathcal{G}^*$ of $\mathcal{G}$ is
$$
\mathcal{G}^*~:~\bm{X}\in\mathbb{C}^{N\times N}\mapsto
\mathcal{G}^*\bm{X}\in\mathbb{C}^{2N-1},\qquad [\mathcal{G}^*\bm{X}]_j=\langle\bm{X},\bm{E}_j\rangle.
$$
Obviously, $\mathcal{G}^*\mathcal{G}$ is the identity operator in $\mathbb{C}^{2N-1}$, and $\mathcal{G}\mathcal{G}^*$ is the orthogonal projector onto the subspace of all Hankel matrices.

\subsection{Recovery condition based on restricted minimum gain condition}
First of all, let us simplify the minimization problem \eqref{eq:min} by introducing $\mathcal{D}\in\mathbb{C}^{(2N-1)\times (2N-1)}$, the diagonal matrix with $j$-th diagonal $\sqrt{K_j}$.
Then, by letting $\bm{y}=\mathcal{D}\bm{x}$, \eqref{eq:min} is rewritten as,
\begin{equation}\label{eq:mincomplex}
\min_{\bm{y}}\|\mathcal{G}\bm{y}\|_*\qquad\mbox{subject to}\quad
\mathcal{B}\bm{y}=\bm{b},
\end{equation}
where $\mathcal{B}=\mathcal{A}\mathcal{D}^{-1}$. Similarly, for the noisy case, \eqref{eq:minnoise} is rearranged to
\begin{equation}\label{eq:mincomplexnoisy}
\min_{\bm{y}}\|\mathcal{G}\bm{y}\|_*\qquad\mbox{subject to}\quad
\|\mathcal{B}\bm{y}-\bm{b}\|_2\leq\epsilon.
\end{equation}
By our assumption in Theorem \ref{thm:main}, $\mathcal{B}\in\mathbb{C}^{M\times (2N-1)}$ is a random matrix whose real and imaginary parts are both real-valued random matrices with i.i.d. Gaussian entries of mean $0$ and variance $1$. We will prove $\tilde{\bm{y}}=\mathcal{D}\hat{\bm{x}}$ (respectively $\|\tilde{\bm{y}}-\hat{\bm{y}}\|_2\leq 2\delta/\epsilon$) with dominant probability for problem \eqref{eq:mincomplex} for the noise free case (respectively \eqref{eq:mincomplexnoisy} for the noisy case).

Let the desent cone of $\|\mathcal{G}\cdot\|_*$ at $\hat{\bm{y}}$ be
\begin{equation}\label{eq:tconecomplex}
\mathfrak{T}(\hat{\bm{y}})=\{\lambda\bm{z}~|~\lambda\geq 0, \|\mathcal{G}(\hat{\bm{y}}+\bm{z})\|_*\leq\|\mathcal{G}\hat{\bm{y}}\|_*\}.
\end{equation}
To characterize the recovery condition, we need to use the minimum value of $\frac{\|\mathcal{B}\bm{z}\|_2}{\|\bm{z}\|_2}$ for nonzero $\bm{z}\in\mathfrak{T}(\hat{\bm{y}})$. This quantity is commonly called the {\it minimum gain} of the measurement operator $\mathcal{B}$ restricted on $\mathfrak{T}(\hat{\bm{y}})$\cite{CRPW:FoCM:12}. In particular, if the minimum gain is bounded away from zero, then the exact recovery (respectively approximate recovery) for problem \eqref{eq:mincomplex} (respectively \eqref{eq:mincomplexnoisy}) holds.
\begin{lemma}\label{lem:tangentcone}
Let $\mathfrak{T}(\hat{\bm{y}})$ be defined by \eqref{eq:tconecomplex}.
Assume
\begin{equation}\label{eq:NullSpace}
\min_{\bm{z}\in\mathfrak{T}(\hat{\bm{y}})}\frac{\|\mathcal{B}\bm{z}\|_2}{\|\bm{z}\|_2}\geq\epsilon.
\end{equation}
\begin{enumerate}
\item[(a)]
Let $\tilde{\bm{y}}$ be the solution of \eqref{eq:mincomplex} with $\bm{b}=\mathcal{B}\hat{\bm{y}}$.
Then $\tilde{\bm{y}}=\hat{\bm{y}}$.
\item[(b)]
Let $\tilde{\bm{y}}$ be the solution of \eqref{eq:mincomplexnoisy} with $\|\bm{b}-\mathcal{B}\hat{\bm{y}}\|_2\leq\delta$.
Then $\|\tilde{\bm{y}}-\hat{\bm{y}}\|_2\leq 2\delta/\epsilon$.
\end{enumerate}
\end{lemma}
\begin{proof}
Since (a) is a special case of (b) with $\delta=0$, we prove (b) only. The optimality of $\tilde{\bm{y}}$ implies $\tilde{\bm{y}}-\hat{\bm{y}}\in\mathfrak{T}(\hat{\bm{y}})$. By \eqref{eq:NullSpace}, we have
$$
\|\tilde{\bm{y}}-\hat{\bm{y}}\|_2\leq \frac{1}{\epsilon}\|\mathcal{B}(\tilde{\bm{y}}-\hat{\bm{y}})\|_2
\leq \frac{1}{\epsilon}(\|\mathcal{B}\tilde{\bm{y}}-\bm{b}\|_2+\|\mathcal{B}\hat{\bm{y}}-\bm{b}\|_2)
\leq 2\delta/\epsilon.
$$
\end{proof}

Minimum gain condition is a powerful concept and has been employed in recent recovery results via $\ell_1$ norm minimization, block-sparse vector recovery,  low-rank matrix reconstruction and other atomic norms \cite{CRPW:FoCM:12}.
\subsection{Bound of minimum gain via Gaussian width}
Lemma \ref{lem:tangentcone} requires to estimate the lower bound of $\min_{\bm{z}\in\mathfrak{T}(\hat{\bm{y}})}\frac{\|\mathcal{B}\bm{z}\|_2}{\|\bm{z}\|_2}$. Gordon gave a solution using Gaussian width of a set \cite{Gor:GAFA:88,CRPW:FoCM:12} to estimate the lower bound of minimum gain.
\begin{definition}
The Gaussian width of a set $S\subset \mathbb{R}^p$ is defined as:
$$w(S):=\mathsf{E}_{\bm{\xi}}\left[\sup_{\bm{\gamma}\in S}{\bm{\gamma}^T\bm{\xi}}\right],$$
where $\bm{\xi}\in\mathbb{R}^{p}$ is a random vector of independent zero-mean unit-variance Gaussians.
\end{definition}
Let $\lambda_n$ denote the expected length of a $n$-dimensional Gaussian random vector. Then $\lambda_n=\sqrt{2}\Gamma(\frac{n+1}{2})/\Gamma(\frac{n}{2})$  and it can be tightly bounded as $\frac{n}{\sqrt{n+1}}\leq\lambda_n\leq\sqrt{n}$ \cite{CRPW:FoCM:12}. The following theorem is given in Corollary 1.2 in \cite{Gor:GAFA:88}. It gives a bound on minimum gain for a random map $\bm{\Pi}:\mathbb{R}^p\mapsto \mathbb{R}^n$.
\begin{theorem}[Corollary 1.2 in \cite{Gor:GAFA:88}]\label{thm:restrictedeigenvalue}
Let $\Omega$ be a closed subset of $\{\bm{x}\in\mathbb{R}^p|\|\bm{x}\|_2=1\}$. Let $\bm{\Pi}\in\mathbb{R}^{n\times p}$ be a random matrix with i.i.d. Gaussian entries with mean $0$ and variance $1$. Then, for any $\epsilon>0$,
$$
\mathsf{P}\left(\min_{\bm{z}\in\Omega}\|\bm{\Pi}\bm{z}\|_2\geq \epsilon\right)
\geq 1-e^{-\frac12\left(\lambda_n-w(\Omega)-\epsilon\right)^2},
$$
provided $\lambda_n-w(\Omega)-\epsilon\geq 0$. Here $\frac{n}{\sqrt{n+1}}\leq\lambda_n\leq\sqrt{n}$, and $w(\Omega)$ is the Gaussian width of $\Omega$.
\end{theorem}

By converting the complex setting in our problem to the real setting and using Theorem \ref{thm:restrictedeigenvalue}, we can get the bound of \eqref{eq:NullSpace} in terms of Gaussian width of $\mathfrak{T}_{\mathbb{R}}(\hat{\bm{y}})\cap \mathbb{S}_{\mathbb{R}}^{4N-3}$, where $\mathfrak{T}_{\mathbb{R}}(\hat{\bm{y}})$ is a cone in $\mathbb{R}^{4N-2}$ defined by
 \begin{equation}\label{eq:tconereal}
\mathfrak{T}_{\mathbb{R}}(\hat{\bm{y}})=\left\{\left[\begin{matrix}\bm{\alpha}\cr\bm{\beta}\end{matrix}\right]\Big|~\bm{\alpha}+\imath\bm{\beta}\in\mathfrak{T}(\hat{\bm{y}})\right\}.
\end{equation}
\begin{lemma}\label{lem:Gaussianwidth}
Let the real and imaginary parts of entries of $\mathcal{B}\in\mathbb{C}^{M\times (2N-1)}$ be i.i.d. Gaussian with mean $0$ and variance $1$. Let $\mathfrak{T}_{\mathbb{R}}(\hat{\bm{y}})$ be defined by \eqref{eq:tconereal} and $\mathbb{S}_c^{2N-2}$ be the unit sphere in $\mathbb{C}^{2N-1}$. Then for any $\epsilon>0$,
$$
\mathsf{P}\left(\min_{\bm{z}\in\mathfrak{T}(\hat{\bm{y}})\cap\mathbb{S}_c^{2N-2}}\|\mathcal{B}\bm{z}\|_2\geq\epsilon\right)\geq 1-2e^{-\frac12\left(\lambda_{M}-w(\mathfrak{T}_{\mathbb{R}}(\hat{\bm{y}})\cap\mathbb{S}_{\mathbb{R}}^{4N-3})-\frac{\epsilon}{\sqrt{2}}\right)^2},
$$
where $\mathbb{S}_{\mathbb{R}}^{4N-3}$ is the unit sphere in $\mathbb{R}^{4N-2}$.
\end{lemma}
\begin{proof}
In order to use Theorem \ref{thm:restrictedeigenvalue}, we convert the complex setting in our problem to the real setting in Theorem \ref{thm:restrictedeigenvalue}. We will use Roman letters for vectors and matrices in complex-valued spaces, and Greek letters for real valued ones. Let $\mathcal{B}=\bm{\Phi}+\imath\bm{\Psi}\in\mathbb{C}^{M\times (2N-1)}$, where both $\bm{\Phi}\in\mathbb{R}^{M\times (2N-1)}$ and $\bm{\Psi}\in\mathbb{R}^{M\times (2N-1)}$ are real-valued random matrices whose entries are i.i.d. mean-$0$ variance-$1$ Gaussian. Then, for any $\bm{z}=\bm{\alpha}+\imath\bm{\beta}\in\mathbb{C}^{2N-1}$ with $\bm{\alpha},\bm{\beta}\in\mathbb{R}^{2N-1}$,
\begin{equation*}
\begin{split}
\|\mathcal{B}\bm{z}\|_2&=\|(\bm{\Phi}+\imath\bm{\Psi})(\bm{\alpha}+\imath\bm{\beta})\|_2
=\left\|(\bm{\Phi}\bm{\alpha}-\bm{\Psi}\bm{\beta})+\imath(\bm{\Psi}\bm{\alpha}+\bm{\Phi}\bm{\beta})\right\|_2\cr
&=\left(\left\|\left[\begin{matrix}\bm{\Phi}&-\bm{\Psi}\end{matrix}\right]
\left[\begin{matrix}\bm{\alpha}\cr \bm{\beta}\end{matrix}\right]\right\|_2^2
+\left\|\left[\begin{matrix}\bm{\Psi}&\bm{\Phi}\end{matrix}\right]
\left[\begin{matrix}\bm{\alpha}\cr \bm{\beta}\end{matrix}\right]\right\|_2^2\right)^{1/2}
\end{split}
\end{equation*}
Then
\begin{equation}\label{eq:event1}
\min_{\bm{z}=\bm{\alpha}+\imath\bm{\beta}\in\mathfrak{T}(\hat{\bm{y}})\cap\mathbb{S}_c^{2N-2}}
\left\|\left[\begin{matrix}\bm{\Phi}&-\bm{\Psi}\end{matrix}\right]
\left[\begin{matrix}\bm{\alpha}\cr \bm{\beta}\end{matrix}\right]\right\|_2\geq \epsilon/\sqrt{2},\quad\mbox{and}
\min_{\bm{z}=\bm{\alpha}+\imath\bm{\beta}\in\mathfrak{T}(\hat{\bm{y}})\cap\mathbb{S}_c^{2N-2}}
\left\|\left[\begin{matrix}\bm{\Psi}&\bm{\Phi}\end{matrix}\right]
\left[\begin{matrix}\bm{\alpha}\cr \bm{\beta}\end{matrix}\right]\right\|_2\geq \epsilon/\sqrt{2}
\end{equation}
implies
$$
\min_{\bm{z}\in\mathfrak{T}(\hat{\bm{y}})\cap\mathbb{S}_c^{2N-2}}\|\mathcal{B}\bm{z}\|_2
\geq \epsilon.
$$
Therefore,
$$
\mathsf{P}\left(\min_{\bm{z}\in\mathfrak{T}(\hat{\bm{y}})\cap\mathbb{S}_c^{2N-2}}\|\mathcal{B}\bm{z}\|_2\geq\epsilon\right)
\geq
\mathsf{P}\left(\mbox{\eqref{eq:event1} holds true}\right).
$$
It is easy to see that both $\left[\begin{matrix}\bm{\Phi}&-\bm{\Psi}\end{matrix}\right]$ and $\left[\begin{matrix}\bm{\Psi}&\bm{\Phi}\end{matrix}\right]$ are real-valued random matrices with i.i.d. Gaussian entries of mean $0$ and variance $1$. By Theorem \ref{thm:restrictedeigenvalue},
$$
\mathsf{P}\left(\mbox{\eqref{eq:event1} holds true}\right)
\geq 1-2e^{-\frac12\left(\lambda_{M}-w(\mathfrak{T}_{\mathbb{R}}(\hat{\bm{y}})\cap\mathbb{S}_{\mathbb{R}}^{4N-3})-\frac{\epsilon}{\sqrt{2}}\right)^2},
$$ and therefore we get the desired result.
\end{proof}

\subsection{Estimation of Gaussian width $w(\mathfrak{T}_{\mathbb{R}}(\hat{\bm{y}})\cap\mathbb{S}_{\mathbb{R}}^{4N-3})$}
Denote $\mathfrak{T}_{\mathbb{R}}^{*}(\hat{\bm{y}}))$ be polar cone of $\mathfrak{T}_{\mathbb{R}}(\hat{\bm{y}}))\in\mathbb{R}^{4N-2}$, i.e.,
\begin{equation}\label{def:polar}
\mathfrak{T}_{\mathbb{R}}^{*}(\hat{\bm{y}})
=\{\bm{\delta}\in\mathbb{R}^{4N-2}~|~\bm{\gamma}^T\bm{\delta}\leq 0,~\forall\bm{\gamma}\in\mathfrak{T}_{\mathbb{R}}(\hat{\bm{y}})\}.
\end{equation}
Following the arguments in Proposition 3.6 in \cite{CRPW:FoCM:12}, we obtain
\begin{equation}\label{eq:widthest1}
w(\mathfrak{T}_{\mathbb{R}}(\hat{\bm{y}})\cap\mathbb{S}_{\mathbb{R}}^{4N-3})
=\mathsf{E}\left(\sup_{\bm{\gamma}\in\mathfrak{T}_{\mathbb{R}}(\hat{\bm{y}})\cap\mathbb{S}_{\mathbb{R}}^{4N-3}}\bm{\xi}^T\bm{\gamma}\right)
\leq \mathsf{E}\left(\min_{\bm{\gamma}\in\mathfrak{T}_{\mathbb{R}}^{*}(\hat{\bm{y}})}\|\bm{\xi}-\bm{\gamma}\|_2\right),
\end{equation}
where $\bm{\xi}\in\mathbb{R}^{4N-2}$ is a random vector of i.i.d. Gaussian entries of mean $0$ and variance $1$.
Hence, instead of estimating Gaussian width $w(\mathfrak{T}_{\mathbb{R}}(\hat{\bm{y}})\cap\mathbb{S}_{\mathbb{R}}^{4N-3})$, we bound $\mathsf{E}\left(\min_{\bm{\gamma}\in\mathfrak{T}_{\mathbb{R}}^{*}(\hat{\bm{y}})}\|\bm{\xi}-\bm{\gamma}\|_2\right)$.
For this purpose, let $\mathcal{F}~:~\mathbb{R}^{4N-2}\mapsto\mathbb{R}$ be defined by
\begin{equation}\label{def:F}
\mathcal{F}\left(\left[\begin{matrix}\bm{\alpha}\cr\bm{\beta}\end{matrix}\right]\right)=\|\mathcal{G}(\bm{\alpha}+\imath\bm{\beta})\|_*.
\end{equation}
The following lemma gives us a characterization of $\mathsf{E}\left(\min_{\bm{\gamma}\in\mathfrak{T}_{\mathbb{R}}^{*}(\hat{\bm{y}})}\|\bm{\xi}-\bm{\gamma}\|_2\right)$ in terms of the subdifferential $\partial\mathcal{F}$ of $\mathcal{F}$.

\begin{lemma}\label{lemma:polar}
Let $\mathfrak{T}_{\mathbb{R}}^{*}(\hat{\bm{y}})$ and $\mathcal{F}$ be defined by \eqref{def:polar} and  \eqref{def:F} respectively. Let $\hat{\bm{\omega}}_1, \hat{\bm{\omega}}_2\in\mathbb{R}^{2N-1}$ be the real and imaginary parts of $\hat{\bm{y}}$ respectively and denote $\hat{\bm{\omega}}=\left[\begin{matrix}\hat{\bm{\omega}}_1\cr\hat{\bm{\omega}}_2\end{matrix}\right]$. Then
\begin{equation}\label{eq:dualcone}
\mathfrak{T}_{\mathbb{R}}^{*}(\hat{\bm{y}})
=\mathrm{cone}\left(\partial \mathcal{F}\left(\hat{\bm{\omega}}\right)\right)=\left\{\lambda\bm{\delta}~|~\lambda\geq0,~\mathcal{F}\left(\bm{\gamma}+\hat{\bm{\omega}}\right)
\geq \mathcal{F}\left(\hat{\bm{\omega}}\right)
+\bm{\gamma}^T\bm{\delta},~\forall \bm{\gamma}\in\mathbb{R}^{4N-2}
\right\}.
\end{equation}
\end{lemma}
\begin{proof}
It is observed that $\mathfrak{T}_{\mathbb{R}}(\hat{\bm{y}})$ in  \eqref{eq:tconereal} is the descent cone of the function $\mathcal{F}$
$$
\mathfrak{T}_{\mathbb{R}}(\hat{\bm{y}})=\left\{\delta\bm{\gamma}~|~\delta\geq0,~\mathcal{F}\left(\bm{\gamma}+\hat{\bm{\omega}}\right)
\leq \mathcal{F}\left(\hat{\bm{\omega}}\right)\right\}.
$$
According to Theorem 23.4 in \cite{Roc:BOOK:97}, the cone dual to the descent cone is the conic hull of subgradient, which is exactly \eqref{eq:dualcone}.
\end{proof}

The following lemma gives us an estimation of Gaussian width $w(\mathfrak{T}_{\mathbb{R}}(\hat{\bm{y}})\cap\mathbb{S}_{\mathbb{R}}^{4N-3})$ in terms of $\mathsf{E}(\|\mathcal{G}\bm{g}\|_2)$.
\begin{lemma} \label{lem:boundGaussianWidth}Let $\mathfrak{T}_{\mathbb{R}}(\hat{\bm{y}})$ and $\mathcal{G}$ be defined by \eqref{eq:tconereal} and \eqref{def:G} respectively. Then
$$
w(\mathfrak{T}_{\mathbb{R}}(\hat{\bm{y}})\cap\mathbb{S}_{\mathbb{R}}^{4N-3})
\leq 3\sqrt{R}\cdot\mathsf{E}(\|\mathcal{G}\bm{g}\|_2),
$$
where $\mathsf{E}(\|\mathcal{G}\bm{g}\|_2)$ is the expectation with respect to $\bm{g}\in\mathbb{C}^{2N-1}$. Here $\bm{g}$ is a random vector whose real and imaginary parts are i.i.d. mean-0 and variance-1 Gaussian entries.
\end{lemma}
\begin{proof}
By using \eqref{eq:widthest1} and Lemma \ref{lemma:polar}, we need to find $\partial \mathcal{F}\left(\hat{\bm{\omega}}\right)$ and thus $\mathfrak{T}_{\mathbb{R}}^{*}(\hat{\bm{y}})$. Let $\hat{\bm{\Omega}}_1=\mathcal{G}\hat{\bm{\omega}}_1$ and $\hat{\bm{\Omega}}_2=\mathcal{G}\hat{\bm{\omega}}_2$. Then $\mathcal{G}\hat{\bm{y}}=\hat{\bm{\Omega}}_1+\imath\hat{\bm{\Omega}}_2$. Let a singular value decomposition of the rank-$R$ matrix $\mathcal{G}\hat{\bm{y}}$ be
\begin{equation}\label{eq:SVDcomplexGy}
\mathcal{G}\hat{\bm{y}}=\bm{U}\bm{\Sigma}\bm{V}^*,
\qquad\mbox{with}\quad
\bm{U} = \bm{\Theta}_1+\imath\bm{\Theta}_2,~~\bm{V}=\bm{\Xi}_1+\imath\bm{\Xi}_2,
\end{equation}
where $\bm{\Theta}_1,\bm{\Theta}_2,\bm{\Xi}_1,\bm{\Xi}_2\in\mathbb{R}^{N\times R}$ and $\bm{\Sigma}\in\mathbb{R}^{R\times R}$, and $\bm{U}\in\mathbb{C}^{N\times R}$ and $\bm{V}\in\mathbb{C}^{N\times R}$ satisfies $\bm{U}^*\bm{U}=\bm{V}^*\bm{V}=\bm{I}$. Then, by direct calculation,
\begin{equation}\label{eq:ThetaXi}
\bm{\Theta}\equiv\left[\begin{matrix}
\bm{\Theta}_1 & -\bm{\Theta}_2\cr \bm{\Theta}_2 & \bm{\Theta}_1
\end{matrix}\right]\in\mathbb{R}^{2N\times (2R)}, \qquad
\bm{\Xi}\equiv\left[\begin{matrix}
\bm{\Xi}_1 & -\bm{\Xi}_2\cr \bm{\Xi}_2 & \bm{\Xi}_1
\end{matrix}\right]\in\mathbb{R}^{2N\times (2R)}
\end{equation}
satisfy $\bm{\Theta}^T\bm{\Theta}=\bm{\Xi}^T\bm{\Xi}=\bm{I}$. Moreover, if we define
$\hat{\bm{\Omega}}=\left[\begin{matrix}
\hat{\bm{\Omega}}_1 & -\hat{\bm{\Omega}}_2\cr \hat{\bm{\Omega}}_2 & \hat{\bm{\Omega}}_1
\end{matrix}\right]$, then
\begin{equation}\label{eq:SVDrealOmega}
\hat{\bm{\Omega}}=
\bm{\Theta}\left[\begin{matrix}
\bm{\Sigma} & \cr & \bm{\Sigma}
\end{matrix}\right]
\bm{\Xi}^T
\end{equation}
is a singular value decomposition of the real matrix $\hat{\bm{\Omega}}$, and the singular values $\hat{\bm{\Omega}}$ are those of $\mathcal{G}\hat{\bm{y}}$, each repeated twice. Therefore,
\begin{equation}\label{eq:Freal}
\mathcal{F}\left(\hat{\bm{\omega}}\right)
=\|\mathcal{G}\hat{\bm{y}}\|_*=\|\bm{\Sigma}\|_*=\frac12\|\hat{\bm{\Omega}}\|_*.
\end{equation}
Define a linear operator $\mathcal{E}~:~\mathbb{R}^{4N-2}\mapsto\mathbb{R}^{2N\times 2N}$ by
$$
\mathcal{E}\left(\left[\begin{matrix}\bm{\alpha}\cr\bm{\beta}\end{matrix}\right]\right)
=\left[\begin{matrix}\mathcal{G}\bm{\alpha}&-\mathcal{G}\bm{\beta}\cr
\mathcal{G}\bm{\beta}&\mathcal{G}\bm{\alpha}\end{matrix}\right],
\quad\mbox{with}\quad\bm{\alpha},\bm{\beta}\in\mathbb{R}^{2N-1}.
$$
By \eqref{eq:Freal} and the definition of $\hat{\bm{\Omega}}$, we obtain $\mathcal{F}(\hat{\bm{\omega}})=\frac12\|\mathcal{E}\hat{\bm{\omega}}\|_*$. From convex analysis theory and $\hat{\Omega}=\mathcal{E}\hat{\omega}$, the subdifferential of $\mathcal{F}$ is given by
\begin{equation}\label{eq:subdiffF}
\partial\mathcal{F}(\hat{\bm{\omega}})=\frac12\mathcal{E}^*\partial\|\hat{\bm{\Omega}}\|_*.
\end{equation}
On the one hand, the adjoint $\mathcal{E}^*$ is given by, for any $\bm{\Delta}=\left[\begin{matrix}\bm{\Delta}_{11}&\bm{\Delta}_{12}\cr\bm{\Delta}_{21}&\bm{\Delta}_{22}\end{matrix}\right]\in\mathbb{R}^{2N\times 2N}$ with each block in $\mathbb{R}^{N\times N}$,
\begin{equation}\label{eq:adjE}
\mathcal{E}^*\bm{\Delta}=\left[\begin{matrix}\mathcal{G}^*(\bm{\Delta}_{11}+\bm{\Delta}_{22})\cr\mathcal{G}^*(\bm{\Delta}_{21}-\bm{\Delta}_{12})\end{matrix}\right].
\end{equation}
On the other hand, since \eqref{eq:SVDrealOmega} provides a singular value decomposition of $\hat{\bm{\Omega}}$,
\begin{equation}\label{eq:subdiffnuc}
\partial\|\hat{\bm{\Omega}}\|_*=\left\{\bm{\Theta}\bm{\Xi}^T+\bm{\Delta}~|~
\bm{\Theta}^T\bm{\Delta}=\bm{0},~\bm{\Delta}\bm{\Xi}=\bm{0},~\|\bm{\Delta}\|_2\leq1\right\}.
\end{equation}
Combining \eqref{eq:subdiffF}\eqref{eq:adjE}\eqref{eq:subdiffnuc} and  \eqref{eq:ThetaXi} yields the subdifferential of $\mathcal{F}$ at $\hat{\bm{\omega}}$
$$
\partial\mathcal{F}(\hat{\bm{\omega}})=\left\{\left[\begin{matrix}
\mathcal{G}^*\left(\bm{\Theta}_1\bm{\Xi}_1^T+\bm{\Theta}_2\bm{\Xi}_2^T+\frac{\bm{\Delta}_{11}+\bm{\Delta}_{22}}{2}\right)\cr
\mathcal{G}^*\left(\bm{\Theta}_2\bm{\Xi}_1^T-\bm{\Theta}_1\bm{\Xi}_2^T+\frac{\bm{\Delta}_{21}-\bm{\Delta}_{12}}{2}\right)
\end{matrix}\right]~\Big|~
\bm{\Delta}=\left[\begin{matrix}\bm{\Delta}_{11}&\bm{\Delta}_{12}\cr\bm{\Delta}_{21}&\bm{\Delta}_{22}\end{matrix}\right],~\bm{\Theta}^T\bm{\Delta}=\bm{0},~\bm{\Delta}\bm{\Xi}=\bm{0},~\|\bm{\Delta}\|_2\leq1
\right\}.
$$

We are now ready for the estimation of the Gaussian width. Let the set $\mathfrak{S}$ be a subset of the set of complex-valued vectors
\begin{equation}\label{eq:setS}
\mathfrak{S}=\left\{\mathcal{G}^*(\bm{U}\bm{V}^*+\bm{W})~|~\bm{U}^*\bm{W}=\bm{0},~\bm{W}\bm{V}=\bm{0},
~\|\bm{W}\|_2\leq 1\right\},
\end{equation}
where $\bm{U},\bm{V}$ are in \eqref{eq:SVDcomplexGy}. Then, it can be checked that
\begin{equation}\label{eq:GsubsetsubdiffF}
\mathfrak{H}\equiv \left\{\left[\begin{matrix}\bm{\alpha}\cr\bm{\beta}\end{matrix}\right]~\Big|~
\bm{\alpha}+\imath\bm{\beta}\in\mathfrak{S}\right\}
\subset \partial\mathcal{F}(\hat{\bm{\omega}}).
\end{equation}
Actually, for any $\bm{W}=\bm{\Delta}_1+\imath\bm{\Delta}_2$ satisfying $\bm{U}^*\bm{W}=0,\bm{W}\bm{V}=0$ and $\|\bm{W}\|_2\leq 1$, we choose $\bm{\Delta}=\left[\begin{matrix}\bm{\Delta}_1&-\bm{\Delta}_2\cr\bm{\Delta}_2&\bm{\Delta}_1\end{matrix}\right]$. Obviously, this choice of $\bm{\Delta}$ satisfies the constraints on $\bm{\Delta}$ in $\partial\mathcal{F}(\hat{\bm{\omega}})$. Furthermore, $\bm{U}\bm{V}^*+\bm{W}=(\bm{\Theta}_1\bm{\Xi}_1^T+\bm{\Theta}_2\bm{\Xi}_2^T+\bm{\Delta}_1) +\imath(\bm{\Theta}_2\bm{\Xi}_1^T+\bm{\Theta}_1\bm{\Xi}_2^T+\bm{\Delta}_2)$. Therefore, \eqref{eq:GsubsetsubdiffF} holds.

With the help of \eqref{eq:GsubsetsubdiffF}, we get
\begin{equation}\label{eq:widthest1.5}
\min_{\bm{\gamma}\in\mathfrak{T}_{\mathbb{R}}^{*}(\hat{\bm{y}})}\|\bm{\xi}-\bm{\gamma}\|_2
=\min_{\lambda\geq 0}\min_{\bm{\gamma}\in\partial\mathcal{F}(\hat{\bm{\omega}})}\|\bm{\xi}-\lambda\bm{\gamma}\|_2
\leq\min_{\lambda\geq 0}\min_{\bm{\gamma}\in\mathfrak{H}}\|\bm{\xi}-\lambda\bm{\gamma}\|_2.
\end{equation}
We then convert the real-valued vectors to complex-valued vectors by letting $\bm{g}=\bm{\xi}_1+\imath\bm{\xi}_2$ and $\bm{c}=\bm{\gamma}_1+\imath\bm{\gamma}_2$, where $\bm{\xi}_1$ and $\bm{\xi}_2$ are the first and second half of $\bm{\xi}$ respectively and so for $\bm{\gamma}_1$ and $\bm{\gamma}_2$. This leads to
$$
\min_{\bm{\gamma}\in\mathfrak{T}_{\mathbb{R}}^{*}(\hat{\bm{y}})}\|\bm{\xi}-\bm{\gamma}\|_2
\leq\min_{\lambda\geq 0}\min_{\bm{\gamma}\in\mathfrak{H}}\|\bm{\xi}-\lambda\bm{\gamma}\|_2
=\min_{\lambda\geq 0}\min_{\bm{c}\in\mathfrak{S}}\|\bm{g}-\lambda\bm{c}\|_2.
$$
Since $\mathcal{G}^*\mathcal{G}$ is the identity operator and $\mathcal{G}\mathcal{G}^*$ is an orthogonal projector, for any $\lambda\geq0$ and $\bm{c}\in\mathfrak{S}$,
\begin{equation}\label{eq:widthest2}
\begin{split}
\|\bm{g}-\lambda\bm{c}\|_2&=\|\mathcal{G}\bm{g}-\lambda\mathcal{G}\bm{c}\|_F
=\|\mathcal{G}\bm{g}-\lambda\mathcal{G}\mathcal{G}^*(\bm{U}\bm{V}^*+\bm{W})\|_F\cr
&=\left(\|\mathcal{G}\bm{g}-\lambda(\bm{U}\bm{V}^*+\bm{W})\|_F^2-\|\lambda(\mathcal{I}-\mathcal{G}\mathcal{G}^*)(\bm{U}\bm{V}^*+\bm{W})\|_F^2\right)^{1/2}\cr
&\leq\|\mathcal{G}\bm{g}-\lambda(\bm{U}\bm{V}^*+\bm{W})\|_F,
\end{split}
\end{equation}
where $\bm{W}$ satisfies the conditions in the definition of $\mathfrak{S}$ in \eqref{eq:setS}. Define two orthogonal projectors $\mathcal{P}_1$ and $\mathcal{P}_2$ in $\mathbb{C}^{N\times N}$ by
$$
\mathcal{P}_1\bm{X}=\bm{U}\bm{U}^*\bm{X}+\bm{X}\bm{V}\bm{V}^*-\bm{U}\bm{U}^*\bm{X}\bm{V}\bm{V}^*,\qquad
\mathcal{P}_2\bm{X}=(\bm{I}-\bm{U}\bm{U}^*)\bm{X}(\bm{I}-\bm{V}\bm{V}^*).
$$
Then, it can be easily checked that: $\mathcal{P}_1\bm{X}$ and $\mathcal{P}_2\bm{X}$ are orthogonal, $\bm{X}=\mathcal{P}_1\bm{X}+\mathcal{P}_2\bm{X}$, and
\begin{equation}\label{eq:P1P2}
\mathcal{P}_1\bm{U}\bm{V}^*=\bm{U}\bm{V}^*,\quad \mathcal{P}_2\bm{W}=\bm{0},\quad
\mathcal{P}_1\bm{W}=\bm{0},\quad \mathcal{P}_2\bm{W}=\bm{W},
\end{equation}
where $\bm{U},\bm{V},\bm{W}$ the same as those in \eqref{eq:setS}. We choose
$$
\lambda=\|\mathcal{P}_2(\mathcal{G}\bm{g})\|_2,\qquad
\bm{W}=\frac{1}{\lambda}\mathcal{P}_2(\mathcal{G}\bm{g}).
$$
Then, $\bm{W}$ satisfies constraints in \eqref{eq:setS}. This, together with \eqref{eq:widthest1.5}\eqref{eq:widthest2}\eqref{eq:P1P2}, implies
\begin{equation*}
\begin{split}
\min_{\bm{\gamma}\in\mathfrak{T}_{\mathbb{R}}^{*}(\hat{\bm{y}})}\|\bm{\xi}-\bm{\gamma}\|_2
&\leq \big\|\mathcal{G}\bm{g}-\|\mathcal{P}_2(\mathcal{G}\bm{g})\|_2\bm{U}\bm{V}^*-\mathcal{P}_2(\mathcal{G}\bm{g})\big\|_F
=\big\|\mathcal{P}_1(\mathcal{G}\bm{g})-\|\mathcal{P}_2(\mathcal{G}\bm{g})\|_2\bm{U}\bm{V}^*\big\|_F\cr
&\leq\|\mathcal{P}_1(\mathcal{G}\bm{g})\|_F+\|\mathcal{P}_2(\mathcal{G}\bm{g})\|_2\|\bm{U}\bm{V}^*\|_F
=\|\mathcal{P}_1(\mathcal{G}\bm{g})\|_F+\sqrt{R}\|\mathcal{P}_2(\mathcal{G}\bm{g})\|_2.
\end{split}
\end{equation*}
We will estimate both $\|\mathcal{P}_1(\mathcal{G}\bm{g})\|_F$ and $\|\mathcal{P}_2(\mathcal{G}\bm{g})\|_2$. For $\|\mathcal{P}_1(\mathcal{G}\bm{g})\|_F$, we have
\begin{equation*}
\begin{split}
\|\mathcal{P}_1(\mathcal{G}\bm{g})\|_F&=\|\bm{U}\bm{U}^*(\mathcal{G}\bm{g})+(\mathcal{G}\bm{g})\bm{V}\bm{V}^*-\bm{U}\bm{U}^*(\mathcal{G}\bm{g})\bm{V}\bm{V}^*\|_F
=\|\bm{U}\bm{U}^*(\mathcal{G}\bm{g})+(\bm{I}-\bm{U}\bm{U}^*)(\mathcal{G}\bm{g})\bm{V}\bm{V}^*\|_F\cr
&\leq \|\bm{U}\bm{U}^*(\mathcal{G}\bm{g})\|_F+\|(\bm{I}-\bm{U}\bm{U}^*)(\mathcal{G}\bm{g})\bm{V}\bm{V}^*\|_F
\leq \|\bm{U}\bm{U}^*(\mathcal{G}\bm{g})\|_F+\|(\mathcal{G}\bm{g})\bm{V}\bm{V}^*\|_F\cr
&\leq 2\sqrt{R}\|\mathcal{G}\bm{g}\|_2
\end{split}
\end{equation*}
where in the last line we have used the inequality
$$
\|\bm{U}\bm{U}^*(\mathcal{G}\bm{g})\|_F\leq\|\bm{U}\bm{U}^*\|_F\|\mathcal{G}\bm{g}\|_2\leq \sqrt{R}\|\mathcal{G}\bm{g}\|_2
$$
and similarly $\|(\mathcal{G}\bm{g})\bm{V}\bm{V}^*\|_F\leq \sqrt{R}\|\mathcal{G}\bm{g}\|_2$.
For $\|\mathcal{P}_2(\mathcal{G}\bm{g})\|_2$,
$$
\|\mathcal{P}_2(\mathcal{G}\bm{g})\|_2=\|(\bm{I}-\bm{U}\bm{U}^*)(\mathcal{G}\bm{g})(\bm{I}-\bm{V}\bm{V}^*)\|_2\leq\|\bm{I}-\bm{U}\bm{U}^*\|_2\|\mathcal{G}\bm{g}\|_2\|\bm{I}-\bm{V}\bm{V}^*\|_2
\leq\|\mathcal{G}\bm{g}\|_2.
$$
Altogether, we obtain
$$
\min_{\bm{\gamma}\in\mathfrak{T}_{\mathbb{R}}^{*}(\hat{\bm{y}})}\|\bm{\xi}-\bm{\gamma}\|_2
\leq 3\sqrt{R}\|\mathcal{G}\bm{g}\|_2,
$$
which together with \eqref{eq:widthest1} gives
$$
w(\mathfrak{T}_{\mathbb{R}}(\hat{\bm{y}})\cap\mathbb{S}_{\mathbb{R}}^{4N-3})
\leq 3\sqrt{R}\cdot\mathsf{E}(\|\mathcal{G}\bm{g}\|_2).
$$
\end{proof}

\subsection{Bound of $\mathsf{E}(\|\mathcal{G}\bm{g}\|_2)$}
The estimation of $\mathsf{E}(\|\mathcal{G}\bm{g}\|_2)$ plays an important role in proving Theorem \ref{thm:main} since it needed to give the tight bound of the Gaussian width  $w(\mathfrak{T}_{\mathbb{R}}(\hat{\bm{y}})\cap\mathbb{S}_{\mathbb{R}}^{4N-3})$. The following Theorem gives us a bound for  $\mathsf{E}(\|\mathcal{G}\bm{g}\|_2)$.
\begin{theorem}\label{lem:randomHankel}
Let $\bm{g}\in\mathbb{R}^{2N-1}$ be a random vector whose entries are i.i.d. Gaussian random variables with mean $0$ and variance $1$, or $\bm{g}\in\mathbb{C}^{2N-1}$ a random vector whose real part and imaginary part have i.i.d. Gaussian random entries with mean $0$ and variance $1$. Then,
$$
\mathsf{E}(\|\mathcal{G}\bm{g}\|_2)\leq C_1\ln N,
$$
where $C_1$ are some positive universal constants.
\end{theorem}

 The proof of Theorem \ref{lem:randomHankel} is relatively complicated. In order to help the reader easily understand the proof, we begin with real case and introduce some ideas and lemmas first.  Assume $\bm{g}\in\mathbb{R}^{2N-1}$ has i.i.d standard Gaussian entries with mean $0$ and variance $1$. Notice that  $\mathcal{G}\bm{g}$ is symmetric. Therefore, for any even integer $k$, $(\mathrm{tr}\left(\mathcal{G}\bm{g}\right)^k)^{1/k}$  is the $k$-norm of vector of singular values, which implies $\|\mathcal{G}\bm{g}\|_2\leq(\mathrm{tr}\left(\mathcal{G}\bm{g}\right)^k)^{1/k}$. This together with Jensen's inequality,
\begin{equation}\label{bound1}
\mathsf{E}(\|\mathcal{G}\bm{g}\|_2)\leq \mathsf{E}\left((\mathrm{tr}\left(\mathcal{G}\bm{g}\right)^k)^{1/k}\right)\leq
\left(\mathsf{E}(\mathrm{tr}\left(\mathcal{G}\bm{g}\right)^k)\right)^{1/k}.
\end{equation}
Thus, in order to get an upper bound of $\mathsf{E}(\|\mathcal{G}\bm{g}\|_2)$, we estimate $\mathsf{E}\left(\mathrm{tr}\left(\left(\mathcal{G}\bm{g}\right)^k\right)\right)$. Denote $\bm{M}=\mathcal{G}\bm{g}$. It is easy to see that
\begin{equation}\label{eq:EMk}
\mathsf{E}(\mathrm{tr}(\bm{M}^k))=\sum_{0\leq i_1,i_2,\ldots,i_k\leq N-1} \mathsf{E}(M_{i_1i_2}M_{i_2i_3}\ldots M_{i_{k-1}i_k}M_{i_ki_1}).
\end{equation}
Therefore, we only need to estimate $\sum_{0\leq i_1,i_2,\ldots,i_k\leq N-1} \mathsf{E}(M_{i_1i_2}M_{i_2i_3}\ldots M_{i_{k-1}i_k}M_{i_ki_1})$.

To simplify the notation, we denote $i_{k+1}=i_1$. Notice that $M_{ij}=\frac{g_{i+j}}{\sqrt{K_{i+j}}}$, where $g_{i+j}$ is a random Gaussian variable and $K_j$  is defined in \eqref{eq:Ej}. Hence, $M_{i_\ell,i_{\ell+1}}=M_{i_{\ell'},i_{\ell'+1}}$ if and only if $i_\ell+i_{\ell+1}=i_{\ell'}+i_{\ell'+1}$. In order to utilize this property, we would like to introduce a graph for any given index $i_1,i_2,\ldots,i_k$ and its equivalent edges on the graph. More specifically, we construct graph $\mathfrak{F}_{i_1,i_2,\ldots,i_k}$ with nodes to be $i_1,i_2,\ldots,i_k$ and edges to be $(i_1,i_2), (i_2,i_3),\ldots,(i_{k-1},i_k),(i_k,i_1)$. Let the weight for the edge $(i_{\ell},i_{\ell+1})$ be $i_{\ell}+i_{\ell+1}$. The edges with the same weights are considered as an equivalent class. Obviously, $M_{i_{\ell},i_{\ell+1}}=M_{i_{\ell'},i_{\ell'+1}}$ if and only if $(i_{\ell},i_{\ell+1})$ and $(i_{\ell'},i_{\ell'+1})$ are in the same equivalent class. Assume there are $p$ equivalent classes of the edges of $\mathfrak{F}_{i_1,i_2,\ldots,i_k}$. These equivalent classes are indexed by $1,2,\ldots,p$ according to their order in the graph traversal $i_1\to i_2\to\ldots\to i_k\to i_1$. We associate the graph $\mathfrak{F}_{i_1,i_2,\ldots,i_k}$ a sequence $c_1c_2\ldots c_k$, where $c_j$ is the index of the equivalent class of the edge $(i_j,i_{j+1})$. We call $c_1c_2\ldots c_k$ the label for the equivalent classes of the graph $\mathfrak{F}_{i_1,i_2,\ldots,i_k}$.

The label for the equivalent classes of the graph $\mathfrak{F}_{i_1,i_2,\ldots,i_k}$ plays an important role in bounding $\mathsf{E}(\|\mathcal{G}\bm{g}\|_2)$. In order to help the reader understand this concept better, we give two specific examples here. For $N=6,k=6,i_1=1,i_2=4,i_3=1,i_4=3,i_5=1,i_6=4$, we have a corresponding graph and its label for the equivalent classes of the graph is $112211$. For $N=6,k=6, i_1=2,i_2=3,i_3=2,i_4=4,i_5=2,i_6=3$, the  label for the equivalent classes of the corresponding graph is $112211$ as well. Therefore, there may be several different index sequences $i_1i_2\ldots i_k$ that correspond to the same label for the equivalent classes of the corresponding graph. Let $\mathfrak{A}_{c_1c_2\ldots c_k}$ be the set of indices whose label of equivalent class of the corresponding graph is $c_1c_2\ldots c_k$, i.e.
\begin{equation}\label{eq:frakA}
\mathfrak{A}_{c_1c_2\ldots c_k}=\{i_1i_2\ldots i_k|\mbox{ the label for the equivalent class of the graph } \mathfrak{F}_{i_1,i_2,\ldots,i_k} \mbox{ is } c_1c_2\ldots c_k\}
\end{equation}

For given $c_1c_2\ldots c_k$, $\mathfrak{A}_{c_1c_2\ldots c_k}$ is a subset of $\{i_1i_2\ldots i_k|i_j\in\{0,1,\ldots,N-1\},~\forall~j=1,\ldots,k\}$. The following lemma gives us an estimate for the bound $\sum_{i_1i_2\ldots i_k\in\mathfrak{A}_{c_1c_2\ldots c_k}} \mathsf{E}(M_{i_1i_2}M_{i_2i_3}\ldots M_{i_{k-1}i_k}M_{i_ki_1})$.
\begin{lemma} Let $\zeta$ be the Riemann zeta function and $\mathfrak{A}_{c_1c_2\ldots c_k}$ be defined in \eqref{eq:frakA}. Define $B(s)=\ln (N+1)$ if $s=2$ and $B(s)=\zeta(s/2)\leq \pi^2/6$ for $s\geq 4$.
Then
\begin{equation}\label{eq:ECs1sp}
\sum_{i_1i_2\ldots i_k\in \mathfrak{A}_{c_1c_2\ldots c_k}}\mathsf{E}(M_{i_1i_2}M_{i_2i_3}\ldots M_{i_{k-1}i_k}M_{i_ki_1})
\leq N \prod_{\ell=1}^{p}B(s_\ell)(s_{\ell}-1)!!
\end{equation}
where $p$ is the number of equivalent classes shown in $c_1c_2\ldots c_k$, and $s_{\ell}$, $\ell=1,\ldots,p$, is the frequency of $\ell$ in $c_1c_2\ldots c_k$.
\end{lemma}
\begin{proof}
We begin with finding free indices for any $i_1,i_2,\ldots,i_k$ in the set $\mathfrak{A}_{c_1c_2\ldots c_k}$. Let $(j_1,j_2)$ be the first edge of the class $1$. Therefore, the weight of the first class is $j_1+j_2$. For convenience, we define $k_1(j_1)=j_1$. The first edge of the class $2$ must have a vertex $k_2(j_1,j_2)$, depending on $j_1$ and $j_2$, and a free vertex, denoted by $j_3$. The weight of the second class is $k_2(j_1,j_2)+j_3$. Similarly, the first edge in class $3$ has a vertex $k_3(j_1,j_2,j_3)$ and a free vertex $j_4$, and the weight is $k_3(j_1,j_2,j_3)+j_4$, and so on. Finally, the first edge in class $p$ has a vertex $k_{p}(j_1,j_2,\ldots,j_p)$ and a free vertex $j_{p+1}$, and the weight is $k_{p}(j_1,j_2,\ldots,j_p)+j_{p+1}$. Recall that the entry $M_{ij}$ is $\frac{g_{i+j}}{\sqrt{K_{i+j}}}$, where $g_{i+j}$ is a random Gaussian variable. Therefore, for any $i_1i_2\ldots i_k\in\mathfrak{A}_{c_1c_2\ldots c_k}$,
\begin{equation}\label{eq:Ei1ik2}
\mathsf{E}(M_{i_1i_2}M_{i_2i_3}\ldots M_{i_{k-1}i_k}M_{i_ki_1})
=\prod_{\ell=1}^{p} \frac{1}{K_{m_{\ell}}^{s_{\ell}/2}}\mathsf{E}\left(g_{m_\ell}^{s_{\ell}}\right),
\end{equation}
where $m_{\ell}=k_{\ell}(j_1,j_2,\ldots,j_{\ell})+j_{\ell+1}$. Therefore, it is non-vanishing if and only if $s_1,s_2,\ldots,s_p$ are all even. In these cases,
\begin{equation}\label{eq:Ei1ik}
\mathsf{E}(M_{i_1i_2}M_{i_2i_3}\ldots M_{i_{k-1}i_k}M_{i_ki_1})=\prod_{\ell=1}^{p} \frac{(s_{\ell}-1)!!}{K_{m_{\ell}}^{s_{\ell}/2}}.
\end{equation}
Summing \eqref{eq:Ei1ik} over $\mathfrak{A}_{c_1c_2\ldots c_k}$, we obtain
\begin{equation*}
\begin{split}
&\sum_{i_1i_2\ldots i_k\in \mathfrak{A}_{c_1c_2\ldots c_k}}\mathsf{E}(M_{i_1i_2}M_{i_2i_3}\ldots M_{i_{k-1}i_k}M_{i_ki_1})
\leq\sum_{j_1=0}^{N-1}\sum_{j_2=0}^{N-1}\ldots\sum_{j_p=0}^{N-1}\sum_{j_{p+1}=0}^{N-1}\prod_{\ell=1}^{p} \frac{(s_{\ell}-1)!!}{K_{m_{\ell}}^{s_{\ell}/2}}\cr
&=\sum_{j_1=0}^{N-1}\sum_{j_2=0}^{N-1}\left(\frac{(s_{1}-1)!!}{K_{k_1(j_1)+j_2}^{s_{1}/2}}
\sum_{j_3=0}^{N-1}\left(\frac{(s_{2}-1)!!}{K_{k_2(j_1,j_2)+j_3}^{s_{2}/2}}
\sum_{i_4=0}^{N-1}\left(\ldots\sum_{j_{p+1}=0}^{N-1}\frac{(s_{p}-1)!!}{K_{k_{p}(j_1,\ldots,j_{p})+j_{p+1}}^{s_{p}/2}}\right)\ldots\right)\right)
\end{split}
\end{equation*}
Since, for any $0\leq c\leq N-1$,
$$
\sum_{\ell=0}^{N-1}\frac{1}{K_{c+\ell}^{s/2}}\leq
\begin{cases}
1+1/2+1/3+\ldots 1/N\leq \ln (N+1)& s=2,\cr
1+1/2^{s/2}+\ldots+1/N^{s/2}\leq \zeta(s/2)& s=4,6,\ldots
\end{cases}
$$
where $\zeta$ is the Riemann zeta function. By defining $B(s)=\ln (N+1)$ if $s=2$ and $B(s)=\zeta(s/2)\leq \pi^2/6$ for $s\geq 4$, the desired result easily follows.
\end{proof}

The desired bound for $\mathsf{E}(\|\mathcal{G}\bm{g}\|_2)$ can be obtained if we know how many different sets of $\mathfrak{A}_{c_1c_2\ldots c_k}$  available in the set $\{i_1i_2\ldots i_k|i_j\in\{0,1,\ldots,N-1\},~\forall~j=1,\ldots,k\}$. Let $\mathfrak{B}_{s_1s_2\ldots s_p}$ be the set of all labels of $p$ equivalent classes with $\ell$-th class containing $s_{\ell}$ equivalent edges respectively, i.e.
\begin{equation}
\mathfrak{B}_{s_1s_2\ldots s_p}=\left\{c_1c_2\ldots c_p\big|\begin{array}{ll}&c_1c_2\ldots c_p \hbox{ is a valid label of equivalent classes in graph } \mathfrak{F}_{i1i2\ldots i_k} \\&\hbox{and there are } s_{\ell} \hbox{ $\ell$'s in the label } c_1c_2\ldots c_p \end{array}\right\}
\end{equation}
Let $\mathfrak{C}_p$ be the set of all possible set of all possible choice of $p$ positive even numbers $s_1,\ldots,s_p$ satisfying $s_1+s_2+\ldots+s_p=k$. Then
\begin{equation}\label{eq:Esumall}
\begin{split}
\mathsf{E}(\mathrm{tr}(\bm{M}^k))&=\sum_{0\leq i_1,i_2,\ldots,i_k\leq N-1} E(M_{i_1i_2}M_{i_2i_3}\ldots M_{i_{k-1}i_k}M_{i_ki_1})\cr
&\leq \sum_{p=1}^{k/2}\sum_{s_1\ldots s_p\in\mathfrak{C}_{p}}\sum_{c_1c_2\ldots c_k\in\mathfrak{B}_{s_1s_2\ldots s_p}}\sum_{i_1i_2\ldots i_k\in \mathfrak{A}_{c_1c_2\ldots c_k}}\mathsf{E}(M_{i_1i_2}M_{i_2i_3}\ldots M_{i_{k-1}i_k}M_{i_ki_1})\cr
\end{split}
\end{equation}
By bounding the cardinality of $\mathfrak{B}_{s_1s_2\ldots s_p}$ and $\mathfrak{C}_p$, we can derive the bound $\mathsf{E}(\mathrm{tr}(\bm{M}^k))$ hence $\mathsf{E}(\|\mathcal{G}\bm{g}\|_2)$ for the real case. The complex case can be proved by directly using the results for the real case. Now, we are in position to prove Theorm \ref{lem:randomHankel}.

\begin{proof}[ Proof of Theorem \ref{lem:randomHankel}]
 Following \eqref{eq:Esumall}, we need to count the cardinality of $\mathfrak{B}_{s_1s_2\ldots s_p}$. For any $c_1c_2\ldots c_k\in\mathfrak{B}_{s_1s_2\ldots s_p}$, we must have $c_1=1$. Therefore, there are $k-1 \choose s_1-1$ choices of the positions of remaining $1$'s in $c_1c_2\ldots c_k$. Once positions for $1$'s are fixed, the position of the first $2$ has to be the first available slot, we have $k-s_1-1 \choose s_2-1$ choices for the positions of remaining $2$'s, and so on.  Thus,
\begin{equation*}
\begin{split}
|\mathfrak{B}_{s_1s_2\ldots s_p}|&\leq{k-1 \choose s_1-1}\cdot{k-s_1-1 \choose s_2-1}\cdot\ldots\cdot
{k-s_1-\ldots-s_{p-1}-1 \choose s_p-1}\cr
&=\frac{(k-1)(k-2)\ldots(k-s_1+1)}{(s_1-1)!}\frac{(k-s_1-1)\ldots(k-s_1-s_2+1)}{(s_2-1)!}
\ldots 1\cr
&=\frac{(k-1)!}{\prod_{\ell=1}^{p}(s_{\ell}-1)!\prod_{\ell=1}^{p-1}(k-s_1-\ldots-s_{\ell})},
\end{split}
\end{equation*}
which together with \eqref{eq:ECs1sp} implies, for any $s_1s_2\ldots s_p\in\mathfrak{C}_p$,
\begin{equation}\label{eq:EBA}
\begin{split}
&\sum_{c_1c_2\ldots c_k\in\mathfrak{B}_{s_1s_2\ldots s_p}}\sum_{i_1i_2\ldots i_k\in \mathfrak{A}_{c_1c_2\ldots c_k}}\mathsf{E}(M_{i_1i_2}M_{i_2i_3}\ldots M_{i_{k-1}i_k}M_{i_ki_1})\cr
\leq& N \frac{(k-1)!}{\prod_{\ell=1}^{p}(s_\ell-2)!!\prod_{\ell=1}^{p-1}(k-s_1-\ldots-s_{\ell})}\prod_{\ell=1}^{p}B(s_\ell)
\end{split}
\end{equation}
Summing \eqref{eq:EBA} over $\mathfrak{C}_{p}$ yields
\begin{equation}\label{eq:EBAC}
\begin{split}
&\sum_{s_1\ldots s_p\in\mathfrak{C}_{p}}\sum_{c_1c_2\ldots c_k\in\mathfrak{B}_{s_1s_2\ldots s_p}}\sum_{i_1i_2\ldots i_k\in \mathfrak{A}_{c_1c_2\ldots c_k}}\mathsf{E}(M_{i_1i_2}M_{i_2i_3}\ldots M_{i_{k-1}i_k}M_{i_ki_1})\cr
\leq& N(k-1)! \sum_{{s_1\ldots s_p\in \mathfrak{C}_{p}}}\frac{\prod_{\ell=1}^{p}B(s_\ell)}{\prod_{\ell=1}^{p}(s_\ell-2)!!\prod_{\ell=1}^{p-1}(k-s_1-\ldots-s_{\ell})}
\end{split}
\end{equation}
Let us estimate the sum in the last line. Let $s$ be the number of $2$'s in $s_1s_2\ldots s_p$. Then, \begin{equation}\label{eq:prodBs}
\prod_{\ell=1}^{p}B(s_\ell)\leq \ln^s(N+1)\left(\frac{\pi^2}{6}\right)^{p-s}.
\end{equation}
Since each $s_1,\ldots,s_p\geq 2$ and there are $p-s$ terms greater than $4$ among them, we have
\begin{equation}\label{eq:prodBs1}
\prod_{\ell=1}^{p}(s_\ell-2)!!\geq 2^{p-s}
\end{equation}
and $k-s_1-\ldots-s_{\ell}=s_{\ell+1}+\ldots+s_p\geq2(p-\ell)$, which implies
\begin{equation}\label{eq:prodBs2}
\prod_{\ell=1}^{p-1}(k-s_1-\ldots-s_{\ell})\geq\prod_{\ell=1}^{p-1}2(p-\ell)=2^{p-1}(p-1)!.
\end{equation}
There are $p\choose s$ choices of the positions of the $s$ $2$'s. Moreover, once the $s$ $2$'s in $s_1s_2\ldots s_p$ are chosen, there are at most
$$
\left(\frac{k}{2}-s\right)\cdot\left(\frac{k}{2}-s-1\right)\cdot\ldots\cdot
\left(\frac{k}{2}-s-(p-s+1)\right)\leq\left(\frac{k}{2}\right)^{p-s}
$$
choices of the remaining $p-s$ $s_j$'s. Altogether,
\begin{equation}\label{eq:EBAC2}
\begin{split}
&\sum_{s_1\ldots s_p\in\mathfrak{C}_{p}}\sum_{c_1c_2\ldots c_k\in\mathfrak{B}_{s_1s_2\ldots s_p}}\sum_{i_1i_2\ldots i_k\in \mathfrak{A}_{c_1c_2\ldots c_k}}\mathsf{E}(M_{i_1i_2}M_{i_2i_3}\ldots M_{i_{k-1}i_k}M_{i_ki_1})\cr
\leq& N(k-1)!\sum_{s=0}^{p}{p\choose s}\left(\frac{k}{2}\right)^{p-s}\ln^s(N+1)\left(\frac{\pi^2}{6}\right)^{p-s}\frac{1}{2^{p-s}2^{p-1}(p-1)!}\cr
=&2N(k-1)!\frac{1}{(p-1)!}\sum_{s=0}^{p}{p\choose s}\left(\frac{k}{2}\right)^{p-s}\ln^s(N+1)\left(\frac{\pi^2}{6}\right)^{p-s}\frac{1}{4^{p-s}2^s}\cr
=&\frac{2N(k-1)!}{(p-1)!}\sum_{s=0}^{p}{p\choose s}\left(\frac{\pi^2k}{48}\right)^{p-s}\left(\frac{\ln (N+1)}{2}\right)^s\cr
=&2N(k-1)!\frac{\left(\frac{\pi^2}{48}k+\frac{\ln (N+1)}{2}\right)^{p}}{(p-1)!}
\end{split}
\end{equation}

Finally, \eqref{eq:EBAC2} is summed over all possible $p$ and we obtain
\begin{equation}\label{eq:Esumal2}
\begin{split}
\mathsf{E}(\mathrm{tr}(\bm{M}^k))&=\sum_{i_1,i_2,\ldots,i_k} \mathsf{E}(M_{i_1i_2}M_{i_2i_3}\ldots M_{i_{k-1}i_k}M_{i_ki_1})\cr
&\leq \sum_{p=1}^{k/2}\sum_{s_1\ldots s_p\in\mathfrak{C}_{p}}\sum_{c_1c_2\ldots c_k\in\mathfrak{B}_{s_1s_2\ldots s_p}}\sum_{i_1i_2\ldots i_k\in \mathfrak{A}_{c_1c_2\ldots c_k}}\mathsf{E}(M_{i_1i_2}M_{i_2i_3}\ldots M_{i_{k-1}i_k}M_{i_ki_1})\cr
&\leq 2N(k-1)!\sum_{p=1}^{k/2}\frac{\left(\frac{\pi^2}{48}k+\frac{\ln (N+1)}{2}\right)^{p}}{(p-1)!}
\end{split}
\end{equation}
By using the fact that, for any $A>0$,
$$
\sum_{p=1}^{k/2}\frac{A^{p}}{(p-1)!}
=A\left(1+A+\frac{A^2}{2!}+\ldots+\frac{A^{k/2-1}}{(k/2-1)!}\right)\leq A e^A,
$$
\eqref{eq:Esumal2} is rearranged into
\begin{equation*}
\begin{split}
\mathsf{E}(\mathrm{tr}(\bm{M}^k))&\leq 2N(k-1)!\left(\frac{\pi^2}{48}k+\frac{\ln (N+1)}{2}\right)e^{\frac{\pi^2}{48}k+\frac{\ln (N+1)}{2}}
=2N\sqrt{N+1}(k-1)!\left(\frac{\pi^2}{48}k+\frac{\ln (N+1)}{2}\right)e^{\frac{\pi^2}{48}k}\cr
&\leq 2(N+1)^{\frac32}k^k\left(\frac{\pi^2}{48}+\frac{\ln (N+1)}{2k}\right)e^{\frac{\pi^2}{48}k}.
\end{split}
\end{equation*}
Let $k$ be the smallest even integer greater than $\frac{24}{\pi^2}\ln (N+1)$. Then using $\|\bm{M}\|_2\leq(\mathrm{tr}(\bm{M}^k))^{1/k}$ lead to
\begin{equation*}
\begin{split}
\mathsf{E}(\|\bm{M}\|_2)&\leq \mathsf{E}((\mathrm{tr}(\bm{M}^k))^{1/k})\leq \left(\mathsf{E}(\mathrm{tr}(\bm{M}^k))\right)^{1/k}
\leq (2(N+1)^{\frac32})^{1/k}k\left(\frac{\pi^2}{48}+\frac{\ln (N+1)}{2k}\right)^{1/k}e^{\frac{\pi^2}{48}}
\cr
&\leq 2^{\frac{\pi^2}{24\ln (N+1)}}\cdot e^{\frac{\pi^2}{16}}\cdot\frac{24}{\pi^2}\ln (N+1)
\cdot \left(\frac{\pi^2}{24}\right)^{\frac{\pi^2}{24\ln (N+1)}}\cdot e^{\frac{\pi^2}{48}}
\leq C_1\ln N,
\end{split}
\end{equation*}
where the constant $C_1$ is some universal constant.

Next, we estimate the complex case. In this case, $\bm{g}\in\mathbb{C}^{2N-1}$, where both its real part and imaginary part have i.i.d. Gaussian entries. Write $\bm{g}=\bm{\xi}+\imath\bm{\eta}$, where $\bm{\xi},\bm{\eta}\in\mathbb{R}^{2N-1}$ are real-valued random Gaussian vectors. From the real-valued case above, we derive
$$
\mathsf{E}(\|\mathcal{G}\bm{\xi}\|_2)\leq C_1\ln N,\quad
\mathsf{E}(\|\mathcal{G}\bm{\eta}\|_2)\leq C_1\ln N.
$$
Therefore,
$$
\mathsf{E}(\|\mathcal{G}\bm{g}\|_2)=\mathsf{E}(\|\mathcal{G}\bm{\xi}+\imath\mathcal{G}\bm{\eta}\|_2)
\leq \mathsf{E}(\|\mathcal{G}\bm{\xi}\|_2)+\mathsf{E}(\|\mathcal{G}\bm{\eta}\|_2)
\leq 2C_1\ln N.
$$
\end{proof}

\subsection{Proof of Theorem \ref{thm:main}}
With Lemmas \ref{lem:tangentcone}, \ref{lem:Gaussianwidth}, \ref{lem:boundGaussianWidth}, and Theorem \ref{lem:randomHankel} in hand, we are in position to prove Theorem \ref{thm:main}.
\begin{proof}[Proof of Theorem \ref{thm:main}]
Since \eqref{eq:mincomplex} is equivalent to \eqref{eq:min} by the relation $\bm{y}=\mathcal{D}\bm{x}$, we only need to prove that $\hat{\bm{y}}=\tilde{\bm{y}}$ for noise free data ( $\|\hat{\bm{y}}-\tilde{\bm{y}}\|_2\leq 2\delta/\epsilon$ for noisy data) with dominant probability. According to Lemma \ref{lem:tangentcone}, we only need to prove \eqref{eq:NullSpace}. By Lemma \ref{lem:Gaussianwidth},
$$
\mathsf{P}\left(\min_{\bm{z}\in\mathfrak{T}(\hat{\bm{y}})\cap\mathbb{S}_c^{2N-2}}\|\mathcal{B}\bm{z}\|_2\geq\epsilon\right)\geq 1-2e^{-\frac12\left(\lambda_{M}-w(\mathfrak{T}_{\mathbb{R}}(\hat{\bm{y}})\cap\mathbb{S}_{\mathbb{R}}^{4N-3})-\frac{\epsilon}{\sqrt{2}}\right)^2}.
$$
Lemma \ref{lem:boundGaussianWidth}, Theorem \ref{lem:randomHankel}, and the inequality $\lambda_M\geq \frac{M}{\sqrt{M+1}}$ imply that
$$\lambda_{M}-w(\mathfrak{T}_{\mathbb{R}}(\hat{\bm{y}})\cap\mathbb{S}_{\mathbb{R}}^{4N-3})-\frac{\epsilon}{\sqrt{2}}\geq \frac{M}{\sqrt{M+1}}-3C_1\sqrt{R}\ln N-\frac{\epsilon}{\sqrt{2}}\geq \sqrt{M-1}-3C_1\sqrt{R}\ln N-\frac{\epsilon}{\sqrt{2}}.$$
When  $M \geq (6C_1\sqrt{R}\ln N+\sqrt{2}\epsilon)^2+1$, we can easily get $\mathsf{P}\left(\min_{\bm{z}\in\mathfrak{T}(\hat{\bm{y}})\cap\mathbb{S}_c^{2N-2}}\|\mathcal{B}\bm{z}\|_2\geq\epsilon\right)\geq 1-2e^{-\frac{M-1}{8}}$. We get the desired result.
\end{proof}


\section{Extension to Structured Low-Rank Matrix Reconstruction}\label{secMatrices}
In this section, we extend our results to low-rank Hankel matrix reconstruction and low-rank Toeplitz matrix reconstruction from their Gaussian measurements.

Since the proof of Theorem \ref{thm:main} does not use the specific property that $\hat{\bm{y}}$ is an exponential signal, Theorem \ref{thm:main} holds true for any low-rank Hankel matrices. We have the following corollary, which reads that any Hankel matrix of size $N\times N$ and rank $R$ can be recovered exactly from its $O(R\ln^2N)$ Gaussian measurements, and this reconstruction is robust to noise.
\begin{corollary}[Low-Rank Hankel Matrix Reconstruction]
Let $\hat{\bm{H}}\in\mathbb{C}^{N\times N}$ be a given Hankel matrix with rank $R$. Let $\hat{\bm{x}}\in\mathbb{C}^{2N-1}$ be satisfying $\hat{x}_{i+j}=\hat{H}_{ij}$ for $0\leq i,j\leq N-1$. Let $\mathcal{A}=\mathcal{B}\mathcal{D}\in\mathbb{C}^{M\times (2N-1)}$, where $\mathcal{B}\in\mathbb{C}^{M\times (2N-1)}$ is a random matrix whose real and imaginary parts are i.i.d. Gaussian with mean $0$ and variance $1$, $\mathcal{D}\in\mathbb{R}^{(2N-1)\times (2N-1)}$ is the same as defined in Theorem \ref{thm:main}. Then, there exists a universal constant $C_1>0$ such that, for any $\epsilon>0$, if
$$
M \geq (C_1\sqrt{R}\ln N+\sqrt{2}\epsilon)^2+1,
$$
then, with probability at least $1-2e^{-\frac{M-1}{8}}$, we have
\begin{enumerate}
\item[(a)]
$\bm{H}(\tilde{\bm{x}})=\hat{\bm{H}}$, where $\tilde{\bm{x}}$ is the unique solution of
$$
\min_{\bm{x}}\|\bm{H}(\bm{x})\|_*\quad\mbox{subject to}\quad \mathcal{A}\bm{x}=\bm{b}
$$
with $\bm{b}=\mathcal{A}\hat{\bm{x}}$;
\item[(b)]
$\|\bm{H}(\tilde{\bm{x}})-\hat{\bm{H}})\|_F\leq 2\delta/\epsilon$, where $\tilde{\bm{x}}$ is the unique solution of
$$
\min_{\bm{x}}\|\bm{H}(\bm{x})\|_*\quad\mbox{subject to}\quad \|\mathcal{A}\bm{x}-\bm{b}\|_2\leq\delta
$$
with $\|\bm{b}-\mathcal{A}\hat{\bm{x}}\|_2\leq\delta$.
\end{enumerate}
\end{corollary}

Moreover, Theorem \ref{thm:main} can be extended to the reconstruction of low-rank Toeplitz matrix from its Gaussian measurements. Let $\hat{\bm{T}}\in\mathbb{C}^{N\times N}$ be a Toeplitz matrix. Let $\hat{\bm{x}}\in\mathbb{C}^{2N-1}$ be a vector satisfying $\hat{x}_{N-1+(i-j)}=\hat{T}_{i,j}$ for $0\leq i,j\leq  N-1$. Let $\bm{P}\in\mathbb{C}^{N\times N}$ be an anti-diagonal matrix with anti-diagonals of $1$. Then, it is easy to check that $\hat{\bm{T}}=\bm{H}(\hat{\bm{x}})\bm{P}$. Thus, we define a linear operator $\bm{T}$ that maps a vector in $\mathbb{C}^ {2N-1}$ to a $N\times N$ Toeplitz matrix by $\bm{T}(\bm{x})=\bm{H}(\bm{x})\bm{P}$. Since $\bm{P}$ is a unitary matrix, one has $\|\bm{T}(\bm{x})\|_*=\|\bm{H}(\bm{x})\bm{P}\|_*=\|\bm{H}(\bm{x})\|_*$. Therefore, the above corollary can be adapted to low-rank Toeplitz matrices. We obtain the following corollary, which states that any Toeplitz matrix of size $N\times N$ and rank $R$ can be recovered exactly from its $O(R\ln^2N)$ Gaussian measurements, and this reconstruction is robust to noise.

\begin{corollary}[Low-Rank Toeplitz Matrix Reconstruction]
Let $\hat{\bm{T}}\in\mathbb{C}^{N\times N}$ be a given Toeplitz matrix with rank $R$. Let $\hat{\bm{x}}\in\mathbb{C}^{2N-1}$ be the vector satisfying $\hat{x}_{N-1+(i-j)}=\hat{T}_{i,j}$ for $0\leq i,j\leq N-1$. Let $\mathcal{A}=\mathcal{B}\mathcal{D}\in\mathbb{C}^{M\times (2N-1)}$, where $\mathcal{B}\in\mathbb{C}^{M\times (2N-1)}$ is a random matrix whose real and imaginary parts are i.i.d. Gaussian with mean $0$ and variance $1$, $\mathcal{D}\in\mathbb{R}^{(2N-1)\times (2N-1)}$ is the same as defined in Theorem \ref{thm:main}. Then, there exists a universal constant $C_1>0$ such that, for any $\epsilon>0$, if
$$
M \geq (C_1\sqrt{R}\ln N+\sqrt{2}\epsilon)^2+1,
$$
then, with probability at least $1-2e^{-\frac{M-1}{8}}$, we have
\begin{enumerate}
\item[(a)]
$\bm{T}(\tilde{\bm{x}})=\hat{\bm{T}}$, where $\tilde{\bm{x}}$ is the unique solution of
$$
\min_{\bm{x}}\|\bm{T}(\bm{x})\|_*\quad\mbox{subject to}\quad \mathcal{A}\bm{x}=\bm{b}
$$
with $\bm{b}=\mathcal{A}\hat{\bm{x}}$;
\item[(b)]
$\|\bm{T}(\tilde{\bm{x}})-\hat{\bm{T}})\|_F\leq 2\delta/\epsilon$, where $\tilde{\bm{x}}$ is the unique solution of
$$
\min_{\bm{x}}\|\bm{T}(\bm{x})\|_*\quad\mbox{subject to}\quad \|\mathcal{A}\bm{x}-\bm{b}\|_2\leq\delta
$$
with $\|\bm{b}-\mathcal{A}\hat{\bm{x}}\|_2\leq\delta$.
\end{enumerate}
\end{corollary}

\section{Numerical Experiments}\label{secNum}
In this section, we use numerical experiments to demonstrate our result and its performance improvement, compared with the results in \cite{TBSR:TIT:13,CC:TIT:14}. In the numerical experiments, we use superpositions of complex sinusoids as test signals. Note that the application of our result is not limited to such signals but any signals that are superpositions of complex exponentials.

The true signal $\hat{\bm{x}}$ is generated as follows. We choose $N=64$, i.e., the dimension of $\hat{\bm{x}}$ is $127$. The frequencies $f_k$, $k=1,\ldots,R$, are uniformly randomly drawn from the interval $[0,1]$. The arguments of the coefficients $c_k$, $k=1,\ldots,R$, are from the interval $[0,2\pi]$ uniformly at random, and their amplitudes are generated by $|c_i|=1+10^{0.5m_i}$ where $m_i$ follows the uniform distribution on $[0,1]$. Then, we synthesize the true signal $\hat{\bm{x}}$ by $\hat{x}_t=\sum_{k=1}^{R}c_ke^{\imath 2\pi f_k t}$ for $t=0,1,\ldots,126$.  For each fixed $M$ and $R$, we test $100$ runs. We plot in Fig. \ref{fig:experiments}(a) the rate of successful reconstruction by \eqref{eq:min}, which is solved by alternating direction method of multipliers (ADMM). We see from Fig. \ref{fig:experiments}(a) that the phase transition of our method is very sharp.

For comparison, we plot the phase transitions of off-the-grid CS \cite{TBSR:TIT:13} and EMaC \cite{CC:TIT:14} in Fig. \ref{fig:experiments}(b) and Fig. \ref{fig:experiments}(c) respectively. These figures are from \cite{CC:TIT:14} under the same setting as ours.
We observe that, for the same $R$, our method generally needs smaller $M$ than off-the-grid CS and EMaC to achieve a high successful reconstruction rate. This illustrates that empirically our method requires fewer measurements than both off-the-grid CS and EMaC for the exact reconstruction of complex sinusoid signals. Finally and importantly, our method does not need a separation condition of frequencies to guarantee a successful recovery.

\begin{figure}
\begin{center}
  \subfigure[Our method: Hankel nuclear norm minimization with random Gaussian projections.]%
    {\includegraphics[width=.33\textwidth]{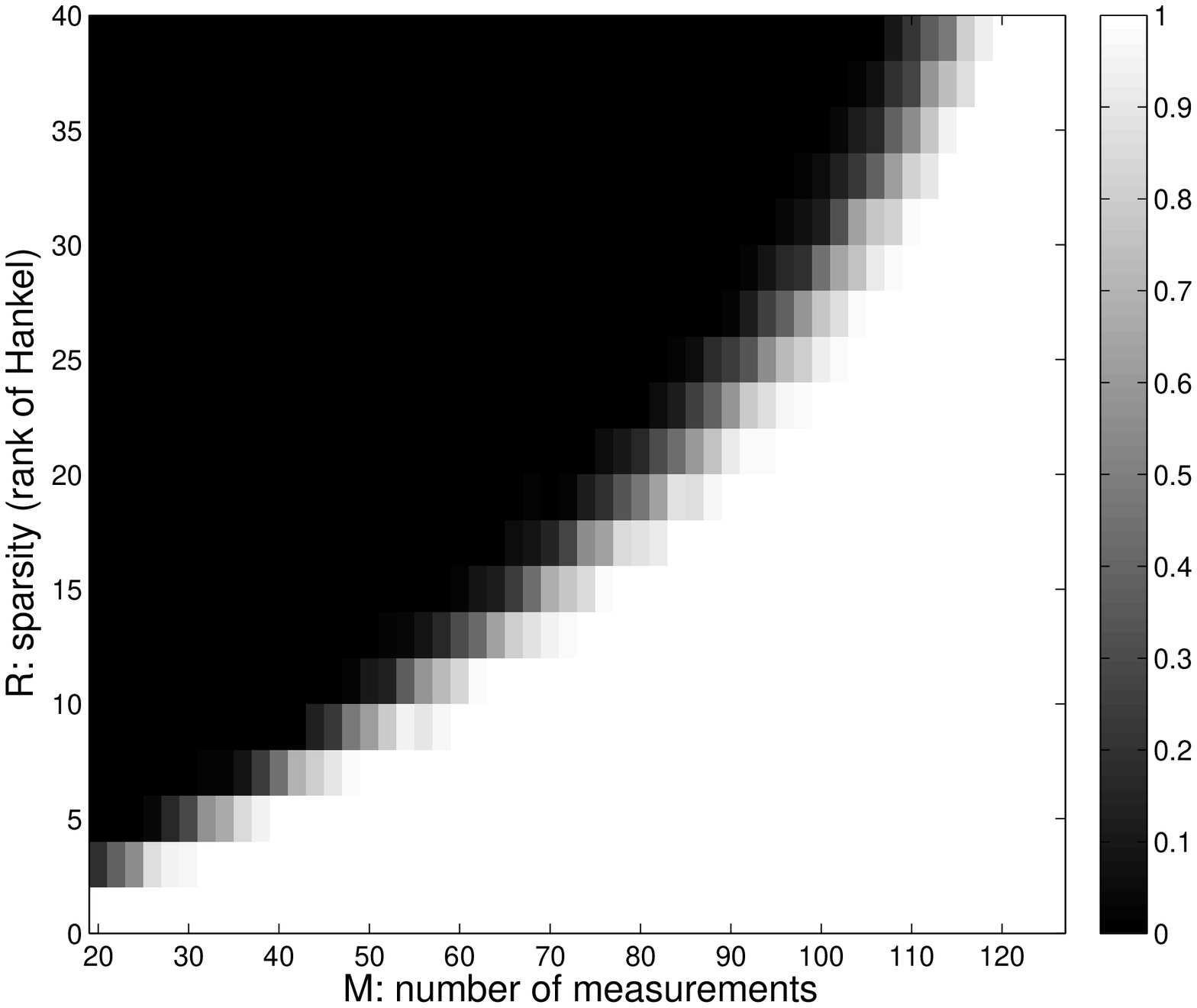}}
  \subfigure[Off-the-grid CS \cite{TBSR:TIT:13}: Atomic norm minimization with non-uniform sampling of entries.]%
    {\includegraphics[width=.31\textwidth]{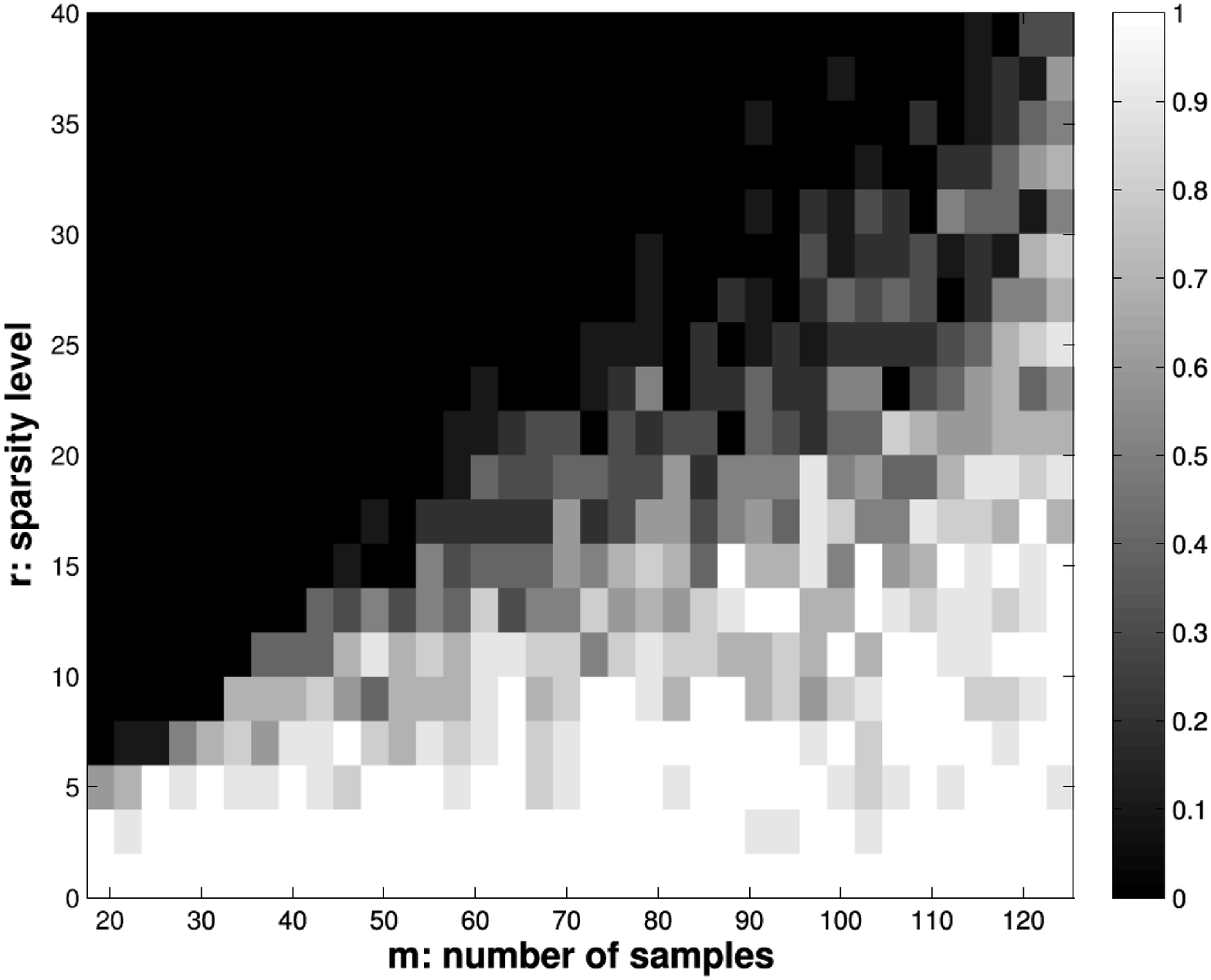}}
  \subfigure[EMaC \cite{CC:TIT:14}: Hankel nuclear norm minimization with non-uniform sampling of entries.]%
    {\includegraphics[width=.31\textwidth]{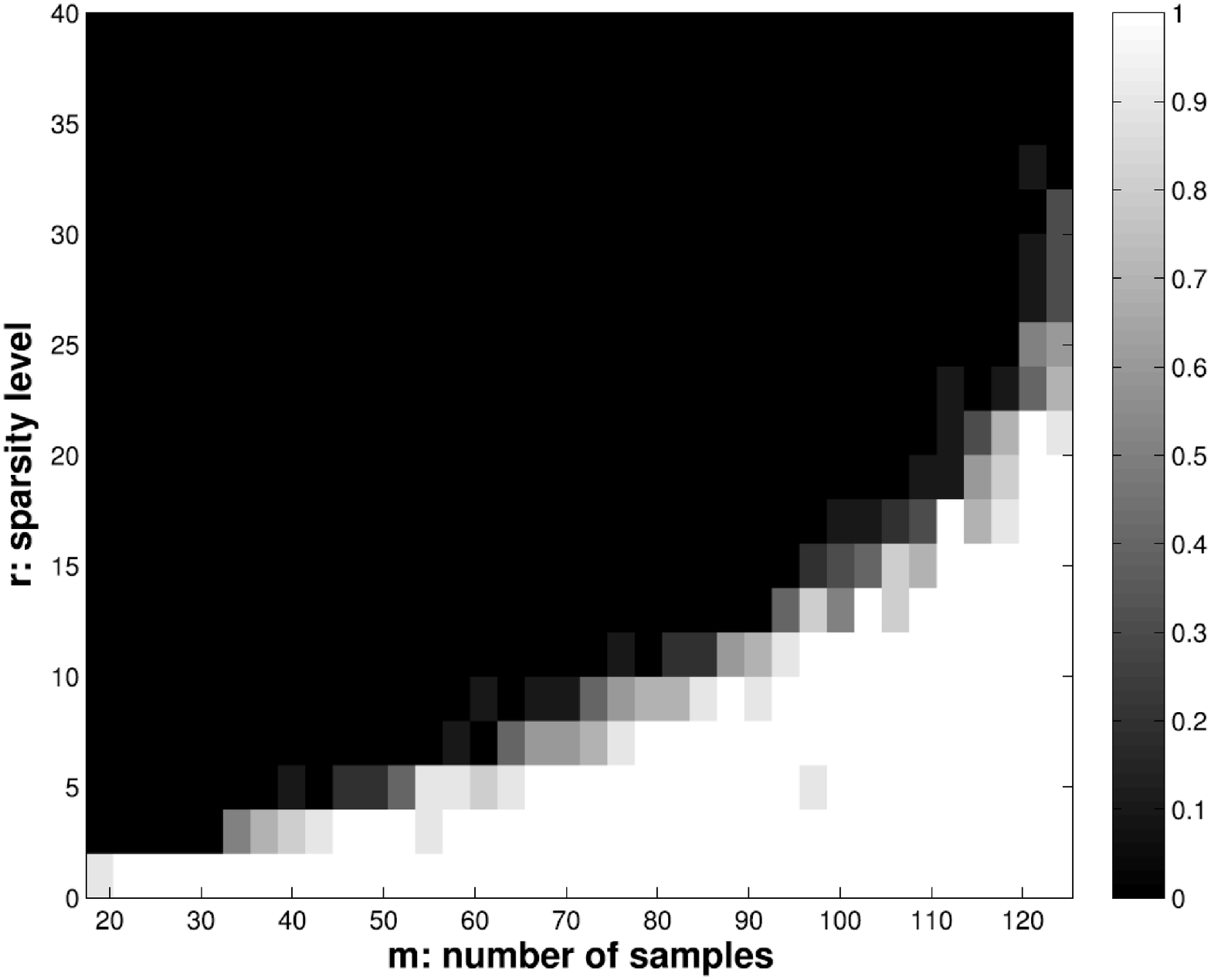}}
\end{center}
\caption{Numerical Results.}\label{fig:experiments}
\end{figure}


\def\cprime{$'$}

\end{document}